\documentclass[aps,pra,reprint,longbibliography]{revtex4-2}

\usepackage[T1]{fontenc}
\usepackage[utf8]{inputenc}
\usepackage{lmodern}
\usepackage{microtype}
\usepackage{mathtools}
\usepackage{amsmath,amssymb,amsthm,amsfonts}
\usepackage{bm}
\usepackage{bbm}
\usepackage{xcolor}
\usepackage{overpic}
\usepackage[colorlinks = true, 
citecolor = blue,
linkcolor = blue,
urlcolor = blue]{hyperref}

\usepackage{kbordermatrix}

\usepackage{tikz} 
\usetikzlibrary{quantikz2}

\usepackage{algorithm}  
\usepackage[noend]{algpseudocode}  
\newcommand{\algorithmfootnote}[2][\footnotesize]{%
  \let\old@algocf@finish\@algocf@finish%
  \def\@algocf@finish{\old@algocf@finish%
    \leavevmode\rlap{\begin{minipage}{\linewidth}
    #1#2
    \end{minipage}}
  }
}

\usepackage{enumitem}
\setlist[itemize,1]{%
  leftmargin=1.5em,
  itemsep=0.5ex,
  topsep=1ex,
  labelsep=0.5em
}
\setlist[itemize,2]{%
  leftmargin=1.5em,
  itemsep=0ex,
  topsep=0pt,
  labelsep=0.5em
}
\newcommand{\bl}[1]{\item[\textcolor{blue}{#1}]\ }

\DeclareMathOperator{\tr}{Tr}
\DeclareMathOperator{\pr}{Pr}
\DeclarePairedDelimiter\ceil{\lceil}{\rceil}

\newcommand{\ind}{\mathbbm{1}}
\newcommand{\E}{\mathbb{E}}
\newcommand{\ketbra}[2]{{\vert #1 \rangle \langle #2 \vert}}
\newcommand{\coleq}{\mathrel{\mathop:}\nobreak\mkern-1.2mu=}

\newcommand{\varep}{\varepsilon}
\newcommand{\bo}{\boldsymbol{o}}
\newcommand{\btheta}{\boldsymbol{\theta}}

\newcommand{\epcue}{\varep\textnormal{-CUE}}
\newcommand{\epuni}{\varep\textnormal{-uniform}}

\newtheorem{theorem}{Theorem}
\newtheorem{lemma}{Lemma}

\allowdisplaybreaks[1]

\begin{document}

\title{On the query complexity of unitary channel certification}

\author{Sangwoo Jeon}
\email{sangw077@gmail.com}
\affiliation{Department of Physics, Korea Advanced Institute of Science and Technology, Daejeon 34141, Korea}

\author{Changhun Oh}
\email{changhun0218@gmail.com}
\affiliation{Department of Physics, Korea Advanced Institute of Science and Technology, Daejeon 34141, Korea}

\date{\today}

\begin{abstract}
    Certifying the correct functioning of a unitary channel is a critical step toward reliable quantum information processing.
    In this work, we investigate the query complexity of the unitary channel certification task: testing whether a given $d$-dimensional unitary channel is identical to or $\varepsilon$-far in diamond distance from a target unitary operation.
    We show that incoherent algorithms—those without quantum memory—require $\Omega(d/\varepsilon^2)$ queries, matching the known upper bound.
    In addition, for general quantum algorithms, we prove a lower bound of $\Omega(\sqrt{d}/\varepsilon)$ and present a matching quantum algorithm based on quantum singular value transformation, establishing a tight query complexity of $\Theta(\sqrt{d}/\varepsilon)$.
    On the other hand, notably, we prove that for almost all unitary channels drawn from a natural average-case ensemble, certification can be accomplished with only $\mathcal{O}(1/\varepsilon^2)$ queries.
    This demonstrates an exponential query complexity gap between worst- and average-case scenarios in certification, implying that certification is significantly easier for most unitary channels encountered in practice.
    Together, our results offer both theoretical insights and practical tools for verifying quantum processes.
\end{abstract}

\maketitle

\section{Introduction}

Reliable quantum information processing critically depends on our ability to verify that quantum processes behave as intended~\cite{preskill1998reliable}. 
While quantum process tomography can accomplish this task for small-sized quantum devices, as quantum devices scale up in size and complexity for more sophisticated quantum information processing, this verification step becomes increasingly challenging~\cite{fawzi2023incoherent, rosenthal2024average}.
Therefore, it is becoming essential to find an efficient way to certify a quantum process and ultimately to develop an optimal and practical scheme.

\textit{Quantum process certification}—the task of verifying that a quantum process operates correctly—is therefore a central challenge in current quantum information processing. 
From an information-theoretic perspective, extensive research has investigated resources necessary for reliable certification~\cite{montanaro2013survey, eisert2020quantum, kliesch2021theory, fawzi2023incoherent, rosenthal2024average}. 
Meanwhile, from a practical engineering perspective, protocols such as quantum process tomography~\cite{chuang1997prescription, acin2001optimal, altepeter2003ancilla, yang2020optimal, haah2023query} and randomized benchmarking~\cite{emerson2005scalable, knill2008randomized, dankert2009exact, magesan2011scalable} have been developed and implemented.
More recently, quantum channel learning techniques have emerged as a promising approach, as they estimate key error observables without fully reconstructing the entire process, which can significantly reduce the required resources~\cite {chen2022quantum, chen2022exponential, chen2022tight, chen2023efficient, huang2023learning, chen2023unitarity, chen2024tight, oh2024entanglement}.

In many practical applications, a desired quantum process to implement is often described by a unitary channel, which plays the role of quantum gates in quantum computing and the perfect transmission of quantum information in quantum communication, because a unitary channel represents a quantum process under ideal and closed-system conditions.
Therefore, among the certification tasks, \textit{unitary channel certification}—the problem of certifying a unitary channel—is particularly important and practically relevant. 
In addition, recent technological advances in quantum coherence have brought laboratory environments closer to ideal closed-system conditions, further underscoring the practical relevance of unitary channel certification~\cite{park2025passive, salhov2024protecting}.
However, somewhat surprisingly, unitary channel certification remains largely unexplored.
Previous studies on quantum process certification have typically considered noisy environments and have shown that certifying completely positive and trace-preserving (CPTP) channels requires exponentially many channel queries~\cite{fawzi2023incoherent, rosenthal2024average}.

In this work, we investigate the unitary channel certification problem and characterize its query complexity.
We first show that incoherent algorithms—those without quantum memory—require exponentially many queries for certification.
We then show that coherent algorithms—general quantum algorithms with quantum memory—can achieve a quadratic speedup over incoherent algorithms through our query-optimal algorithm based on quantum singular value transformation (QSVT), although coherent algorithms still require exponentially many queries.
On the other hand, we show that this exponential hardness arises only in worst-case scenarios and can be significantly reduced for average-case unitary channels.
In particular, we show that there exists a simple algorithm that certifies almost all unitary channels drawn from a natural average-case ensemble using only a constant number of queries.
These results demonstrate an exponential gap between worst- and average-case query complexities, suggesting that certification is substantially easier in practice than in the worst-case scenario.

We organize our work as follows.
In Sec.~\ref{sec:setup}, we provide a detailed definition of the problem setup for unitary channel certification, along with essential definitions.
In Sec.~\ref{sec:works}, we address relevant prior works and highlight our contribution.
In Sec.~\ref{sec:worst}, we establish the tight query complexity for unitary channel certification, showing that unitary channel certification requires exponentially many queries. 
Conversely, in Sec.~\ref{sec:average}, we show that for almost all unitary channels sampled from an average-case ensemble, the query complexity significantly reduces to a constant number. 
Finally, we summarize our findings and discuss their implications in Sec.~\ref{sec:discussion}.

\section{Problem setup}
\label{sec:setup}

We define unitary channel certification as the task of testing whether a given unitary channel is either identical to or $\varep$-far from a target unitary channel~\cite{eisert2020quantum, kliesch2021theory, fawzi2023incoherent, rosenthal2024average}.
We detail the problem setup below.
Suppose one has black-box access to a given unitary channel $\mathcal{E}_U(\rho)\coleq U\rho U^\dagger$, where $U$ is a $d$-dimensional unitary operator acting on an $n$-qubit system with $d=2^n$.
The given unitary channel $\mathcal{E}_U$ is intended to match a target unitary channel $\mathcal{E}_V$.
However, in practice, systematic imperfections such as cross-talk or gate miscalibration may introduce coherent errors, causing $\mathcal{E}_U$ to deviate from $\mathcal{E}_V$.
Therefore, certification is required to guarantee that we are implementing a desired unitary circuit, using as few queries to $\mathcal{E}_U$ as possible.

We formally define the certification task as follows: \textit{testing whether the channel $\mathcal{E}_U$ is identical to $\mathcal{E}_V$ or $\varepsilon$-far from $\mathcal{E}_V$ using $N$ queries to $\mathcal{E}_U$ with success probability at least $2/3$.}
Here, by applying a unitary transformation of the form $\rho\mapsto V^\dagger\rho V$, we can simplify the task to certifying whether the unitary channel $\mathcal{E}_{UV^\dagger}$ is identical to the identity channel $\mathcal{E}_I$. 
Thus, without loss of generality, we set the target channel to be the identity channel and redefine the certification task as follows: \textit{testing whether the channel $\mathcal{E}_U$ is identical to $\mathcal{E}_I$ or $\varepsilon$-far from $\mathcal{E}_I$ using $N$ queries to $\mathcal{E}_U$ with success probability at least $2/3$.}
Thus, certification can be framed as a hypothesis-testing problem:
\begin{align}
H_0:\mathcal{E}_U=\mathcal{E}_I\quad \text{vs.}\quad H_1:D(\mathcal{E}_U,\mathcal{E}_I)\geq\varepsilon,
\end{align}
with a suitable distance metric $D(\cdot,\cdot)$.
Here, if $0<D(\mathcal{E}_U,\mathcal{E}_I)<\varepsilon$, the algorithm is allowed to output either hypothesis.
We employ the diamond distance as the distance metric:
\begin{align}
D(\mathcal{E}_{U},\mathcal{E}_V)=\max_\rho \| (\mathcal{E}_{U}\otimes\mathcal{E}_{I})(\rho)-(\mathcal{E}_{V}\otimes \mathcal{E}_{I})(\rho) \|_1,
\end{align}
where $\|\cdot\|_1$ denotes the Schatten-1 norm defined by $\|M\|_1=\tr(\sqrt{M^\dagger M})$.
Note that the diamond distance captures the worst-case trace distance between output states over all possible input states~\cite{wilde2013quantum}.

\begin{figure}[tb]
    \centerline{
        \begin{overpic}[width=0.37\textwidth]{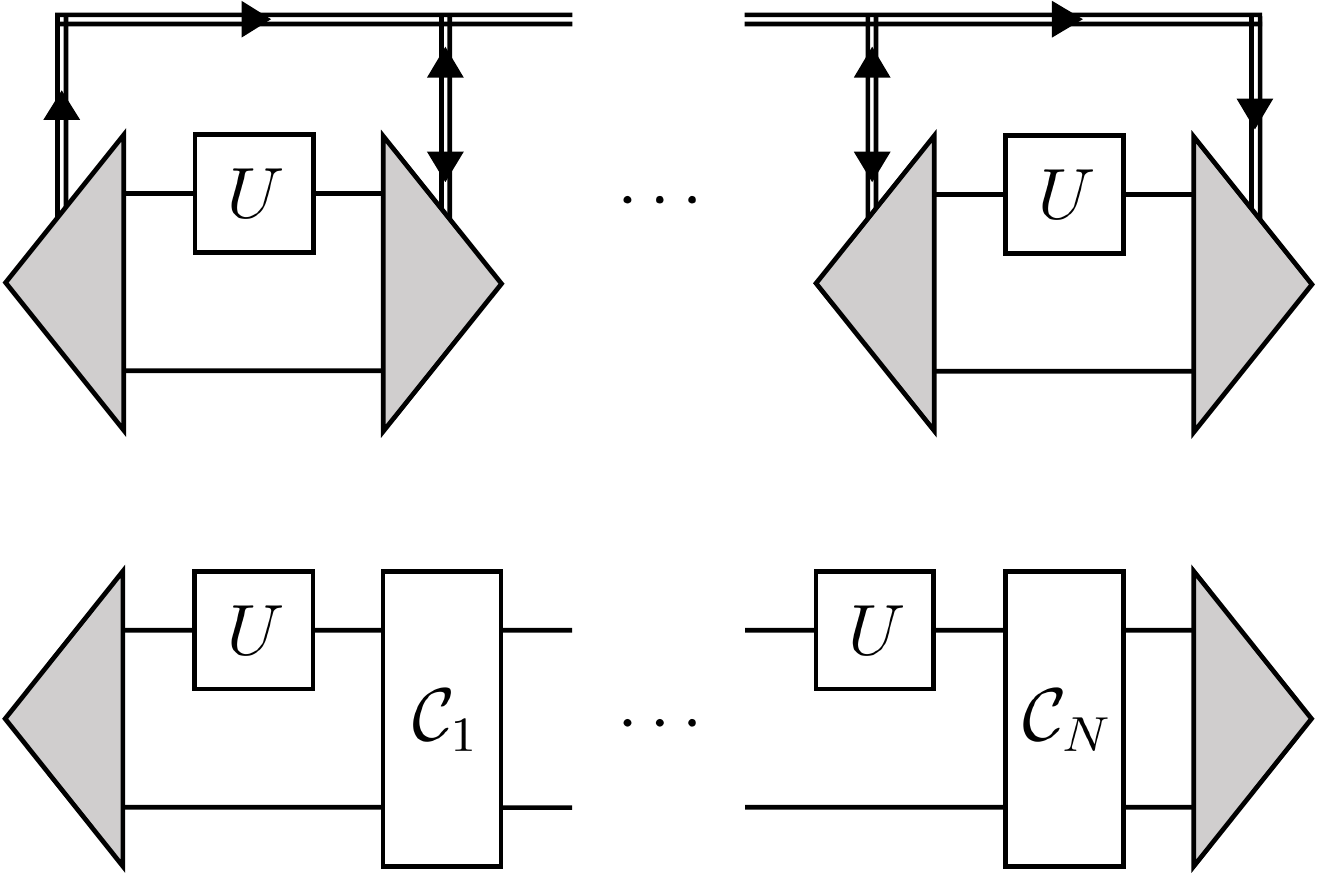}
        \put(-3,68){\text{(a)}} 
        \put(-3,26){\text{(b)}} 
    \end{overpic}
    }
    \caption{
    (a) Incoherent algorithm. The double line represents the classical registers. 
    (b) Coherent algorithm. $\mathcal{C}_k$ for $1\leq k\leq N$ represents a CPTP map that we apply at $k$-th step as part of the algorithm.
    We allow arbitrarily large ancillary systems for both algorithms.
    }
    \label{fig:algorithm}
\end{figure}

We consider two types of algorithms for certification: incoherent and coherent.
Incoherent algorithms, illustrated in Fig.~\ref{fig:algorithm}(a), perform positive operator-valued measurements (POVMs) after each of the $N$ queries.
These algorithms can be adaptive using classical registers to select both input states and POVMs based on previous measurement outcomes.
This approach is practically motivated as storing quantum states across multiple queries in quantum memory is technically challenging.
In contrast, coherent algorithms, illustrated in Fig.~\ref{fig:algorithm}(b), maintain quantum coherence across queries by storing intermediate quantum states.
More specifically, a single input state sequentially passes through $N$ circuit layers, each consisting of the ancilla-coupled unitary channel $\mathcal{E}_U$ and an interleaved CPTP map $\mathcal{C}_k$ for $1\leq k\leq N$.
A final POVM is then performed for certification.
In this work, we extend this conventional framework to cover a wider range of quantum algorithms.
Specifically, we allow arbitrarily large ancillary systems for both types of algorithms.
We also permit the use of the inverse channel $\mathcal{E}_{U^\dagger}$ in place of certain queries to $\mathcal{E}_{U}$, noting that such access is often feasible in practice when $\mathcal{E}_{U}$ is given as a quantum circuit, as reversing the gate sequence and inverting each gate suffices to implement $\mathcal{E}_{U^\dagger}$.
Under these assumptions, coherent algorithms represent the most general class, encompassing incoherent algorithms as a special case.

\section{Background and Contributions}
\label{sec:works}

Let us review prior works to highlight the key contributions of our work in comparison.
The most relevant prior studies are the ones in Refs.~\cite{fawzi2023incoherent} and \cite{rosenthal2024average}, which address the general channel certification problem: certifying whether a given CPTP channel is either identical to or $\varep$-far from a target unitary channel in the diamond distance.
Specifically, Ref.~\cite{fawzi2023incoherent} establishes a tight query complexity of $\Theta(d/\varep^2)$ for incoherent algorithms, while Ref.~\cite{rosenthal2024average} proves a lower bound of $\Omega(\sqrt{d}/\varep)$ for coherent algorithms but does not provide a matching upper bound. 
These results indicate that certifying a CPTP channel in a high-dimensional system is inherently a challenging task.

Recent progress on quantum coherence~\cite{park2025passive, salhov2024protecting} and error correction~\cite{acharya2024quantum} suggests that nearly noiseless quantum processes may be feasible in the near term.
Thus, it is natural and essential to ask whether this hardness persists when the given CPTP channel is restricted to be a unitary channel.
Our first contribution is to show that the same lower bounds hold even under this unitary assumption, \textit{i.e.}, incoherent and coherent algorithms require $\Omega(d/\varep^2)$ and $\Omega(\sqrt{d}/\varep)$ queries for unitary channel certification, respectively, thereby strengthening the previous results.
This result has two major implications:
First, coherent (\textit{i.e.}, unitary) error is a fundamental source of the exponential hardness in quantum process certification.
Second, despite the recent advances in reducing incoherent errors, the exponential hardness of certification remains unavoidable.

Nevertheless, finding an optimal quantum algorithm for certification remains an important challenge.
Our second contribution is to develop a query-optimal certification algorithm for coherent strategies, achieving the tight complexity of $\Theta(\sqrt{d}/\varepsilon)$ by employing QSVT.
This implies that using quantum memory to combine multiple queries coherently can yield a quadratic speedup in certification.

Due to the high complexity of certification under the diamond distance, prior research has attempted to relax the task by considering alternative distance measures.
In particular, these works have employed average-case distances to achieve constant query complexity by avoiding the hardness associated with worst-case instances.
Ref.~\cite{montanaro2013survey} showed a constant query complexity $\mathcal{O}(1/\varep^2)$ for certification under a fidelity-based distance $D(\mathcal{E}_U,\mathcal{E}_V)=\sqrt{1-|\tr (U^\dagger V)|^2/d^2}$, and more recently Ref.~\cite{rosenthal2024average} showed the same query complexity for an average-case imitation diamond distance $D(\mathcal{E}_U,\mathcal{E}_V)=\| (\mathcal{E}_{U}\otimes\mathcal{E}_{I})(\Phi)-(\mathcal{E}_{V}\otimes \mathcal{E}_{I})(\Phi) \|_1$ where $\Phi$ is a maximally entangled state over two $d$-dimensional Hilbert spaces.
Although these average-case results significantly ease the query complexity, their relevance to practical certification remains less clear.

As our last contribution, we show a constant query complexity $\mathcal{O}(1/\varep^2)$ for certification with the diamond distance by considering the average-case \textit{channels}.
We show that there exists a simple algorithm achieving this complexity for almost all unitary channels sampled from a natural average-case distribution.
Here, the fraction of exceptional channels is on the order of $\exp(-\Omega(d))$, which is exponentially small in the system dimension.
This suggests that certification is significantly less challenging in practice than previously believed, offering a highly relevant framework for practical certification.

\section{Worst-case query complexity}
\label{sec:worst}

We now present our main result.
We begin by establishing the query complexity of unitary channel certification in the standard worst-case scenario, \textit{i.e.}, the number of queries required to certify an arbitrary unitary channel.

\subsection{Query complexity for incoherent algorithms}
\label{subsec:incoh}

We prove that certifying unitary channels requires exponentially many queries for incoherent algorithms even when using arbitrarily large ancillary systems and adaptive strategies.
Our result is stated as follows:

\vspace{1em}

\begin{theorem}
    \label{thm:incoh}
    Consider an adaptive, incoherent algorithm with an arbitrarily large ancillary system, which tests whether $D(\mathcal{E}_U,\mathcal{E}_I) \geq \varep$ or $\mathcal{E}_U = \mathcal{E}_I$ with success probability at least $2/3$. For $\varep < 1/2$ and $d > 50\varep^2$, the required number of queries to $\mathcal{E}_U$ (or $\mathcal{E}_{U^\dagger}$) is $N = \Omega(d/\varepsilon^2)$.
\end{theorem}

\noindent
This result strengthens the established lower bound that incoherent algorithms require $\Omega(d/\varep^2)$ queries to certify general CPTP channels~\cite{fawzi2023incoherent}.
Specifically, it shows that the same bound applies even when the given CPTP channel is guaranteed to be unitary.

To prove Theorem \ref{thm:incoh}, we consider a related hypothesis-testing task, which serves as a restricted version of the certification task.
Let $E_\varepsilon$ be an ensemble of $\varepsilon$-perturbed unitary channels $\mathcal{E}_U$, each satisfying $D(\mathcal{E}_U, \mathcal{E}_I) = \varepsilon$.
We consider testing whether $\mathcal{E}_U$ is the identity channel or is sampled from the ensemble $E_\varep$:
\begin{align}
    \label{eq:hypothesis}
    H_0:\mathcal{E}_U=\mathcal{E}_I\quad \text{vs.}\quad H_1:\mathcal{E}_U\sim E_\varep.
\end{align}

\noindent
Since a channel $\mathcal{E}_U$ sampled from $E_\varep$ always satisfies $D(\mathcal{E}_U, \mathcal{E}_I)\geq\varep$ by construction, any algorithm that successfully certifies unitary channels must be able to distinguish these two hypotheses.
Thus, the query complexity of this hypothesis test provides a lower bound on the complexity of the original certification task. 
Therefore, it is sufficient to analyze the query complexity of this problem to derive the lower bound of the unitary channel certification problem.

\begin{figure}[tb]
    \centering
    \begin{overpic}[width=0.41\textwidth]{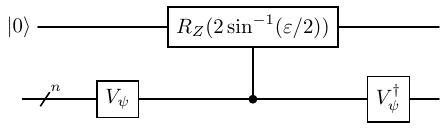}
    \end{overpic}
    \caption{Circuit implementation of a single-basis rotation channel $\mathcal{E}_{U_\psi}$. The unitary operator $V_\psi$ maps the basis state $\ket{\psi}$ to the computational basis state $\ket{1}^{\otimes n}$. The controlled-rotation gate applies a phase shift of $2\sin^{-1}(\varepsilon/2)$ to this computational basis component. As a result, the entire circuit behaves as a phase-shifting channel for the $\ket{\psi}$ basis and as an identity channel for all other bases.}
    \label{fig:worst_ensemble}
\end{figure}

We now construct an ensemble $E_\varep$ to which the corresponding hypothesis testing requires many queries. 
The ensemble we construct is given as follows:
\begin{align}
    E_\varep&=\{\mathcal{E}_{U_\psi}\}_{\ket{\psi}\sim\text{Haar}},\\
    U_\psi&\coleq I+(e^{2i\sin^{-1}(\varep/2)}-1)\ketbra{\psi}{\psi},
\end{align}
where $\ket{\psi}$ is a $d$-dimensional Haar-random state.
Here, each unitary channel $\mathcal{E}_{U_\psi}$ in this ensemble induces a phase shift of $2\sin^{-1}(\varepsilon/2)$ only on the basis $\ket{\psi}$ and acts as the identity elsewhere (see Fig.~\ref{fig:worst_ensemble}).
Reflecting this structure, we refer to $\mathcal{E}_{U_\psi}$ as the \textit{single-basis rotation channel} and to the ensemble $E_\varep$ as the \textit{single-basis rotation ensemble}.
To confirm that $E_\varepsilon$ forms an ensemble of $\varepsilon$-perturbed unitary channels from the identity channel, we examine the structure of the diamond distance $D(\mathcal{E}_U, \mathcal{E}_I)$.
The following lemma expresses it in terms of the eigenangles $\theta_1, \dots, \theta_d$, the arguments of the complex eigenvalues $e^{i\theta_1}, \dots, e^{i\theta_d}$ of the unitary operator $U$:

\begin{lemma}
    \label{lemma:distance} (\cite{haah2023query, kliesch2021theory})
    Let $[\theta_{\text{min}},\theta_{\text{max}}]$ be the shortest interval including all eigenangles of $U$. Then for $\varep<2$, $D(\mathcal{E}_U, \mathcal{E}_I)=\varep$ is equivalent to $\theta_{\text{max}}-\theta_{\text{min}}=2\sin^{-1}(\varep/2)$.
\end{lemma}

\noindent 
Applying this lemma to the channel $\mathcal{E}_{U_\psi}$, only one eigenangle corresponding to the $\ket{\psi}$ basis is nonzero (equal to $2\sin^{-1}(\varepsilon/2)$), while the remaining eigenangles are all zero.
Thus, we have $\theta_{\min}=0$ and $\theta_{\max}=2\sin^{-1}(\varep/2)$, confirming that $\mathcal{E}_{U_\psi}$ is $\varep$-perturbed as $D(\mathcal{E}_{U_\psi}, \mathcal{E}_I)=\varep$; thus, $E_\varep$ is an ensemble of $\varep$-perturbed unitary channels from the identity channel.

Now, we conclude that testing the hypothesis—distinguishing an identity channel from a random channel from $E_\varep$—is exponentially hard for an incoherent algorithm, requiring $\Omega(d/\varep^2)$ queries.
The rest of the proof is outlined in the following proof sketch:

\begin{proof}[Proof sketch of Theorem \ref{thm:incoh}.]
    We employ LeCam's two-point method~\cite{lecam1973convergence} to analyze the hypothesis testing problem defined in Eq.~\eqref{eq:hypothesis}.
    This method relates the testing error probability to the total variation distance (TVD) between the probability distributions of observables under the two hypotheses.
    More specifically, LeCam's method implies that achieving a small testing error requires a sufficiently large TVD between these distributions.
    Thus, we show that a query complexity of $\Omega(d/\varep^2)$ is necessary to obtain such a large TVD.
    This directly implies that the same complexity is required for the certification task.

    Our proof proceeds in two main steps.
    First, we define a suitable \textit{good set} of the measurement outcomes and show that for arbitrary measurements, most outcomes lie within this set, except possibly for a small fraction.
    Next, we show that within this good set, the likelihood ratio between the distributions corresponding to the two hypotheses is concentrated around 1, \textit{i.e.}, the two hypotheses are informationally hard to distinguish.
    To quantify this concentration rigorously, we employ a martingale-based concentration inequality from Ref.~\cite{chen2023efficient}. 
    This step yields an explicit upper bound on the achievable TVD as a function of the number of queries $N$. 
    Together, these results establish the claimed complexity lower bound. 
    The detailed proof is provided in Appendix \ref{appsec:incoh}.
\end{proof}

Note that this lower bound is tight as there exists a matching upper bound established by Ref.~\cite{fawzi2023incoherent}.
Specifically, the following algorithm based on random state preparation and measurement achieves the matching upper bound of $\mathcal{O}(d/\varep^2)$:

\begin{figure}[ht]
\begin{algorithm}[H]
		\caption{Query-optimal incoherent algorithm for unitary channel certification~\cite{fawzi2023incoherent}}
		\label{alg:incoh}
		\begin{algorithmic}[1]
			\Require
			$N$ copies of an $d$-dimensional unitary channel $\mathcal{E}_U$.
			\Ensure
			Decide whether $H_0:\mathcal{E}_U=\mathcal{E}_I$ or $H_1:D(\mathcal{E}_U, \mathcal{E}_I)\geq\varep$.
			\For{$i=1~\textbf{to}~N$}
			\State Input Haar-random $\ket{\psi}$ to $\mathcal{E}_U$.
			\State Measure output with POVM $\{\ketbra{\psi}{\psi}, I-\ketbra{\psi}{\psi}\}$.
            \State Obtain outcome $X_i=0$ or $X_i=1$, respectively.
			\If{$X_i=1$}
			\State\Return Decide $H_1$.
			\EndIf
			\EndFor
            \State\Return Decide $H_0$.
		\end{algorithmic}
\end{algorithm}
\end{figure}

\subsection{Query complexity for coherent algorithms}

In various quantum hypothesis-testing scenarios, jointly measuring multiple queries simultaneously—known as joint measurement—often yields substantial advantages compared to measuring each query individually~\cite{shapiro2020quantum, zhuang2021quantum, coroi2025exponential}. 
Thus, it is valuable to extend our analysis beyond incoherent algorithms and consider general coherent algorithms.

We prove that unitary channel certification requires exponentially many queries, even for coherent algorithms with arbitrarily large ancillary systems. 
This result highlights the fundamental hardness of certification.
Our result is stated as follows:

\vspace{1em}
\begin{theorem}
    \label{thm:coh_lower}
    Consider a coherent algorithm with an arbitrarily large ancillary system, which tests whether $D(\mathcal{E}_U,\mathcal{E}_I)\geq \varep$ or $\mathcal{E}_U=\mathcal{E}_I$ with success probability at least $2/3$. For $\varep<1/2$, the required number of queries to $\mathcal{E}_U$ (or $\mathcal{E}_{U^\dagger}$) is $N=\Omega({\sqrt{d}}/{\varep})$.
\end{theorem}

\noindent 
This strengthens the established lower bound that coherent algorithms require $\Omega(\sqrt{d}/\varep)$ queries to certify general CPTP channels~\cite{rosenthal2024average}.
Specifically, it shows that the same lower bound applies even when the channel is guaranteed to be unitary.
This also generalizes the lower bound for Boolean function certification, which requires $ \Omega(\sqrt{d})$ queries~\cite{montanaro2013survey}.

\begin{proof}[Proof sketch of Theorem \ref{thm:coh_lower}.]
    Consider the output states $\rho_0$ and $\rho_1$ corresponding to hypotheses $H_0$ and $H_1$ in Eq.~\eqref{eq:hypothesis}, respectively.
    The hypothesis-testing error probability is bounded by the trace distance between these two states~\cite{helstrom1969quantum}.
    In coherent algorithms, each pair of an ancilla-coupled channel $\mathcal{E}_U$ and the CPTP map $\mathcal{C}_k$ can increase this trace distance by at most $\mathcal{O}(\varepsilon/\sqrt{d})$, due to the contractivity of trace distance under CPTP maps~\cite{nielsen2010quantum}.
    Therefore, achieving an error probability of at least $2/3$ requires query complexity $\Omega(\sqrt{d}/\varepsilon)$. 
    The detailed proof is provided in Appendix \ref{appsubsec:coh_lower}. 
    We note that the proof is similar to the one given by Ref.~\cite{rosenthal2024average}.
\end{proof}

Theorem \ref{thm:coh_lower} highlights the exponential hardness of certification.
Meanwhile, we observe that if information about the basis state $\ket{\psi}$ associated with each single-basis rotation channel $\mathcal{E}_{U_\psi} \sim E_\varepsilon$ is given, one can certify $\mathcal{E}_{U_\psi}$ using only constant queries of $\mathcal{O}(1/\varepsilon^2)$ via the Hadamard test on the channel $\mathcal{E}_{U_\psi}$ and the state $\ket{\psi}$.
This indicates that the hardness given in Theorem \ref{thm:coh_lower} arises from the unknown information on the phase-rotating basis state $\ket{\psi}$ of $\mathcal{E}_{U_\psi}$.

This type of issue is frequently referred to as \textit{finding a needle in a haystack}, as one has to find a single basis state in a large-dimensional Hilbert space.
A well-known solution to this is Grover's algorithm, which achieves a quadratic speedup over the brute-force approach in a basis-search problem~\cite{grover1996fast, bennett1997strengths}.
Motivated by this, we present a novel Grover-like algorithm achieving the optimal query complexity of $\mathcal{O}(\sqrt{d}/\varep)$, thereby exhibiting a \textit{quadratic speedup} compared to incoherent algorithms.
Our result is stated as follows:

\begin{theorem}
    \label{thm:coh_upper}
    There exists a coherent algorithm that tests whether $\mathcal{E}_U=\mathcal{E}_I$ or $D(\mathcal{E}_U, \mathcal{E}_I)\geq\varep$ with success probability at least $2/3$ using $N=\mathcal{O}(\sqrt{d}/\varep)$ queries to $\mathcal{E}_U$ and $\mathcal{E}_{U^\dagger}$.
\end{theorem}

\noindent
Together with Theorem \ref{thm:coh_lower}, this establishes a tight query complexity of $\Theta(\sqrt{d}/\varep)$ for unitary channel certification with coherent algorithms.
This also implies that allowing quantum memory between queries leads to a quadratic speedup—by a factor of $\Theta(\sqrt{d}/\varep)$—over incoherent algorithms.
We note that access to the inverse channel $\mathcal{E}_{U^\dagger}$ is not a stringent assumption as $U$ is often implemented as a quantum circuit composed of a known sequence of standard gates, in which case $\mathcal{E}_{U^\dagger}$ can be realized by simply reversing the gate sequence and replacing each gate with its inverse.
In addition, the same assumption is also used in Theorems $\ref{thm:incoh}$ and $\ref{thm:coh_lower}$ for a fair comparison.

We provide an intuitive description of our algorithm by comparing it with Grover's algorithm, leaving the full version to the end of the section.
The goal of Grover's algorithm is to search for the bit-flipping basis $\ket{m}$ with an oracle $I-2\ketbra{m}{m}$.
To achieve this, Grover's algorithm amplifies the overlap between an initial superposition state $\ket{s}=(\ket{1}+\dots+\ket{d})/\sqrt{d}$ and the target state $\ket{m}$, using alternating rotations around $\ket{s}$ and $\ket{m}$.
By precisely tuning the number of rotations, one can drive the input state towards the target state $\ket{m}$, thus achieving the searching task.
In contrast, our algorithm performs a process of \textit{amplitude deamplification}, reducing the initially large overlap between two states—a Haar-random state $\ket{\psi}$ and a slightly-rotated $U\ket{\psi}$—to near zero.
More specifically, the algorithm takes a Haar-random input state $\ket{\psi}$ and applies alternating rotations around $\ket{\psi}$ and $U\ket{\psi}$.
Under $H_1$, this drives the state toward a state orthogonal to $\ket{\psi}$, while under $H_0$, the rotations preserve the initial $\ket{\psi}$.
A POVM $\{\ketbra{\psi}{\psi}, I-\ketbra{\psi}{\psi}\}$ then distinguishes between $H_0$ and $H_1$, certifying the unitary channel.

A central challenge in adapting Grover's approach lies in the uncertainty of the appropriate number of rotations.
Grover's algorithm requires a precise number of rotations, which is a fixed value depending on the initial overlap $\braket{s}{m} = 1/\sqrt{d}$.
In our case, the number of rotations depends on the overlap $\bra{\psi}U^\dagger\ket{\psi}$ between $\ket{\psi}$ and $U\ket{\psi}$, which is unknown and varies with both $U$ and the randomly chosen $\ket{\psi}$.
Thus, we cannot directly adopt Grover’s iterative structure.

Therefore, we leverage QSVT, a powerful framework for designing quantum algorithms based on polynomial transformations of operators~\cite{gilyen2019quantum, martyn2021grand}.
We briefly introduce the key concept of QSVT to fully construct our algorithm.
Suppose one has black-box access to a unitary operator $V$ and its inverse $V^\dagger$.
Let $\Pi$ and $\tilde{\Pi}$ be orthogonal projections, and consider the sub-block $S = \Pi V \tilde{\Pi}$ of $V$, which can be expressed in block-encoding form as:
\begin{align}
    V=
    \kbordermatrix{
        &\Pi &  \\
        \tilde{\Pi} & S & \cdot \\ 
        & \cdot & \cdot 
    }.
\end{align}
QSVT enables a polynomial transformation of the singular values of $S$ using $V$, $V^\dagger$, and phase rotations controlled by the projectors $\Pi$ and $\tilde{\Pi}$.
To illustrate, let $S=W\Sigma \tilde{W}^\dagger$ be the singular value decomposition of the sub-block $S$.
Then, QSVT yields a new operator $P^{(\text{SV})}(S)=W P(\Sigma) \tilde{W}^\dagger$ for a real polynomial $P$ satisfying certain conditions.
This leads to the following transformed block encoding:
\begin{align}
    V_\Phi=
    \kbordermatrix{
        &\Pi &  \\
        \tilde{\Pi}\text{ or }\Pi & P^{(\text{SV})}(S) & \cdot \\ 
        & \cdot & \cdot 
    },
\end{align}
where $V_\Phi$ is the result of a QSVT circuit.
The procedure for constructing the QSVT circuit is formally stated in the following lemma:
\begin{lemma} (\cite{gilyen2019quantum})
    \label{lemma:qsvt}
    Let $\Pi$ and $\tilde{\Pi}$ be orthogonal projections and define $\Pi_\phi\coleq e^{i\phi(2\Pi-I)}$ as a projector-controlled phase-rotation gate with angle $\phi$. Suppose $P$ is a real polynomial satisfying:
    \begin{enumerate}[label=(\arabic*)]
        \item $\deg(P)=n$
        \item $P$ shares the same parity as $n$.
        \item $|P(x)|\leq 1$ for $x\in[-1, 1]$.
    \end{enumerate}
    Then, for a given unitary operator $V$, there exist angles $\Phi=(\phi_1,\dots,\phi_n)$ such that the unitary operator
    \begin{align}
        V_\Phi=
        \begin{cases}
            \tilde{\Pi}_{\phi_1}V\prod_{k=1}^{(n-1)/2}\Pi_{\phi_{2k}}V^\dagger\tilde{\Pi}_{\phi_{2k+1}}V & n\text{ is odd}\\
            \prod_{k=1}^{n/2}\Pi_{\phi_{2k-1}}V^\dagger\tilde{\Pi}_{\phi_{2k}}V & n\text{ is even}
        \end{cases}
    \end{align}
    satisfies
    \begin{align}
        P^{(\text{SV})}(\Pi V\tilde{\Pi})=
        \begin{cases}
            \Pi V_\Phi\tilde{\Pi} & n\text{ is odd}\\
            \Pi V_\Phi\Pi & n\text{ is even}
        \end{cases}.
    \end{align}
\end{lemma}
\noindent
Details on determining the rotation angles $\Phi$ from the polynomial $P$ can be found in Ref.~\cite{gilyen2019quantum}.

Collecting the results, we now present the full description of our algorithm.
Our algorithm proceeds in three steps: prepare a Haar-random state $\ket{\psi}$, apply a QSVT operator $V_\Phi$, and perform a POVM $\{\ketbra{\psi}{\psi}, I-\ketbra{\psi}{\psi}\}$.
Following the notation in Lemma \ref{lemma:qsvt}, we construct the operator $V_\Phi$ using projections $\Pi=\ketbra{\psi}{\psi}$ and $\tilde{\Pi}=U\ketbra{\psi}{\psi}U^\dagger$, along with a real polynomial $P$ chosen as a rescaled Chebyshev polynomial.
Under this construction, $V_\Phi$ corresponds to a sequence of alternating rotations around $\ket{\psi}$ and $U\ket{\psi}$ with rotation angles determined by the polynomial $P$.
We show that for almost every Haar-random $\ket{\psi}$, this transformation maps the initial singular value $|\bra{\psi}U^\dagger\ket{\psi}|$ to a transformed singular value $|\bra{\psi}V_\Phi\ket{\psi}|$ that is close to one under $H_0$ and close to zero under $H_1$, without requiring knowledge of the exact overlap between $\ket{\psi}$ and $U\ket{\psi}$.
This ensures that the measurement outcome reliably distinguishes between the two hypotheses, therefore enabling certification of the given channel.
Furthermore, we show that the QSVT circuit $V_\Phi$ can be implemented using $\mathcal{O}(\sqrt{d}/\varep)$ queries to $\mathcal{E}_U$ and $\mathcal{E}_{U^\dagger}$, thereby proving Theorem~\ref{thm:coh_upper}.
The complete proof is provided in Appendix~\ref{appsubsec:coh_upper}, and we summarize the algorithm below:

\begin{figure}[ht]
\begin{algorithm}[H]
        \caption{Query-optimal coherent algorithm for unitary channel certification}
        \label{alg:coh}
        \begin{algorithmic}[1]
            \Require
            Unitary channel $\mathcal{E}_{V_\Phi}$ from QSVT, using $N$ copies of $\mathcal{E}_U$ and $\mathcal{E}_{U^\dagger}$.
            \Ensure
            Decide whether $H_0:\mathcal{E}_U=\mathcal{E}_I$ or $H_1:D(\mathcal{E}_U, \mathcal{E}_I)\geq\varep$.
            \State Input Haar-random $\ket{\psi}$ to $\mathcal{E}_{V_\Phi}$.
            \State Measure output with POVM $\{\ketbra{\psi}{\psi}, I-\ketbra{\psi}{\psi}\}$.
            \State Obtain outcome $M=0$ or $M=1$, respectively.
            \If{$M=0$}
                \State\Return Decide $H_0$.
            \Else
                \State\Return Decide $H_1$.
            \EndIf
        \end{algorithmic}
\end{algorithm}
\end{figure}

\section{Average-case query complexity}
\label{sec:average}

So far, we have established the exponential hardness of unitary channel certification by showing that the identity channel is hard to distinguish from a randomly sampled single-phase rotation channel $\mathcal{E}_{U_\psi}$, where $\ket{\psi}$ is sampled from the Haar measure.  
Here, the channel $\mathcal{E}_{U_\psi}$ can be viewed as a multiqubit-controlled phase rotating operation (see Fig.~\ref{fig:worst_ensemble}), which is highly nonlocal and unlikely to arise under standard local noise models.
This naturally raises the question of whether the exponential hardness we established is overly pessimistic or rarely encountered in practical situations.
Indeed, efficient algorithms for average-case scenarios commonly exist across various quantum testing frameworks, such as quantum channel learning~\cite{huang2021information} and quantum state certification~\cite{huang2024certifying}.
Motivated by these observations, we examine the following question: Can the hardness of certification be relaxed if we consider average-case unitary channels?

To address this question, we first need to clearly define what constitutes the \textit{average case} for random unitary channels.
A conventional and natural choice of a random unitary ensemble is the circular unitary ensemble (CUE), which corresponds to the Haar measure over the unitary group~\cite{dyson1962algebraic, brandao2016local}.
However, in our setting, the CUE itself is not an appropriate notion of average-case unitary channels because the CUE does not adequately represent $\varep$-perturbed unitary channels, and thus fails to offer a fair comparison with the single-basis rotation ensemble $E_\varep$.
For a fair comparison, we must instead consider an ensemble consisting exclusively of $\varep$-perturbed unitary channels.
Thus, we introduce the ensemble $\epcue$, defined as the marginal distribution of the CUE conditioned on the channel being $\varep$-perturbed. 
Precisely, its corresponding measure $\mu_{\epcue}$ is given as:
\begin{align}
    \mu_{\epcue}(A)\coleq \pr_{U\sim{\text{CUE}}}(U\in A|D(\mathcal{E}_U, \mathcal{E}_I)=\varep)
\end{align}
for a set $A$.

We show that for almost every randomly chosen unitary $U\sim\epcue$, except for an exponentially small fraction, there exists a simple, nonadaptive, and ancilla-free algorithm capable of certifying the channel $\mathcal{E}_U$ using only a \textit{constant number of queries}.
Our result is stated as follows:

\vspace{1em}
\begin{theorem}
\label{thm:average}
    Suppose a random unitary channel $\mathcal{E}_U$ is given with $U\sim\epcue$ under $\varep<1/2$ and dimension $d\geq4$. 
    There exists an algorithm that tests whether $D(\mathcal{E}_U, \mathcal{E}_I)\geq\varep$ or $\mathcal{E}_U=\mathcal{E}_I$ with success probability at least 2/3 using $N=\mathcal{O}(1/\varep^2)$ queries, except for $\exp(-\Omega(d))$ fraction of $U$.
\end{theorem}

\begin{figure}[tb]
    \centering
    \begin{overpic}[width=0.37\textwidth]{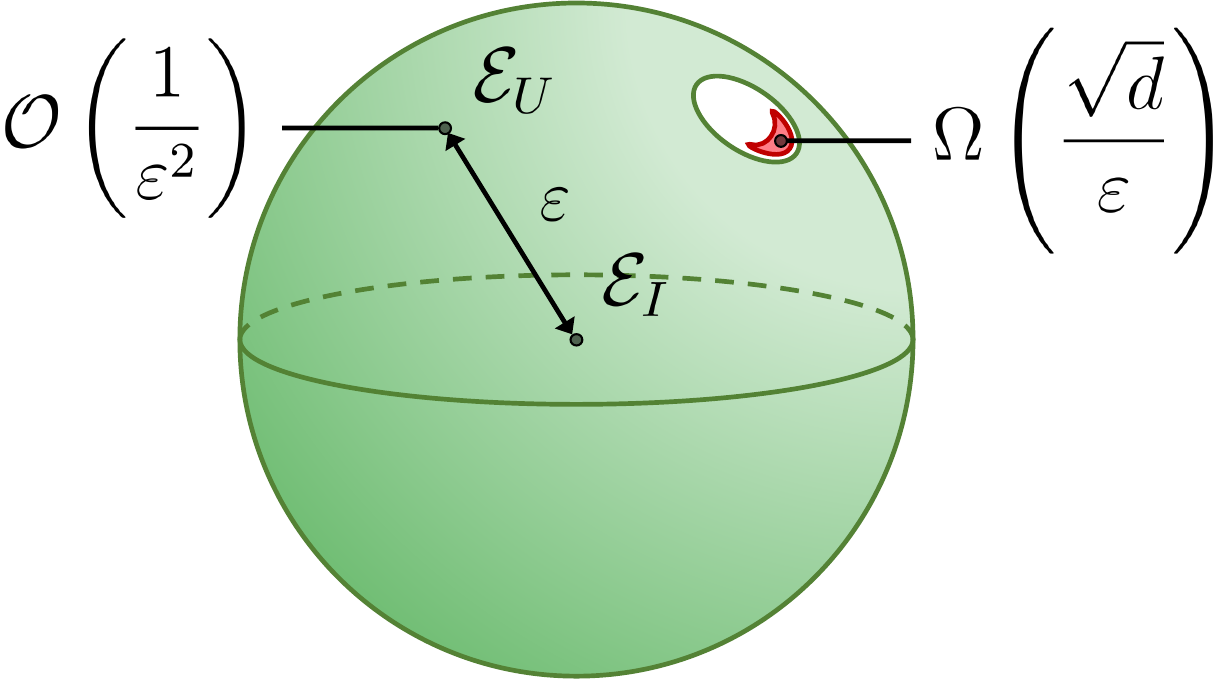} 
    \end{overpic}
    \caption{Visualization of query complexities for $\varep$-perturbed unitary channels. The spherical shell represents the set of $\varep$-perturbed unitary channels sampled from $\epcue$. The green region represents average-case channels, which can be certified using $\mathcal{O}(1/\varep^2)$ queries. The red region represents the single-basis rotation ensemble $E_\varep$, which requires $\Omega(\sqrt{d}/\varepsilon)$ queries for certification. The white region represents a small exceptional subset of measure $\exp(-\Omega(d))$ with unknown complexity.}
    \label{fig:ball}
\end{figure}

\noindent 
Theorem \ref{thm:average} establishes an exponentially large gap between the query complexity of worst-case and average-case scenarios, as illustrated in Fig.~\ref{fig:ball}. 
This emphasizes the importance and practical relevance of considering the average-case scenario in quantum process certification.

\begin{figure*}[htb]
    \centering
    \begin{overpic}[width=0.90\textwidth]{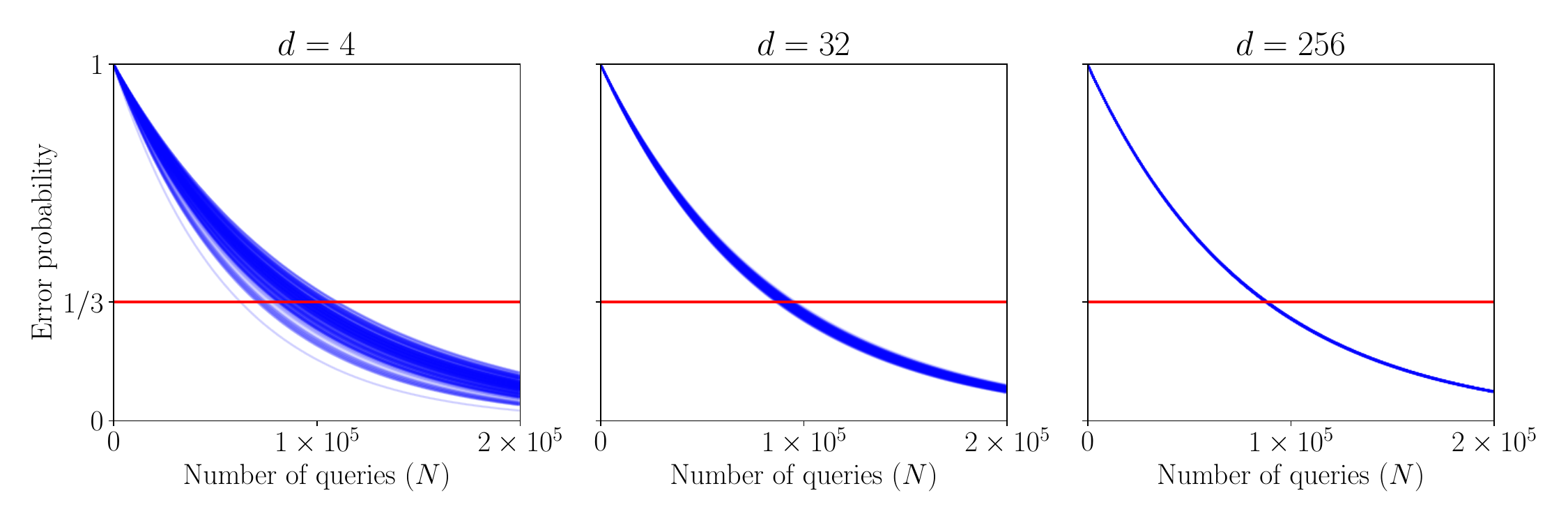}
        \put(4,31){\text{(a)}} 
        \put(35,31){\text{(b)}} 
        \put(66,31){\text{(c)}} 
    \end{overpic}
    \caption{Simulated error probabilities from numerical experiments of Algorithm \ref{alg:incoh}. We plot error probabilities for 200 randomly sampled unitary channels drawn from $\epcue$, with the error threshold $\varep=0.01$ and dimensions (a) $d=4$, (b) $d=32$, and (c) $d=256$. Each blue curve represents the error probability of a single channel as a function of the number of queries $N$. The red horizontal line indicates the targeted error threshold of $1/3$.}
    \label{fig:error}
\end{figure*}

Algorithm~\ref{alg:incoh} introduced in Sec.~\ref{subsec:incoh} achieves the query complexity stated in Theorem~\ref{thm:average}.
We point out that the algorithm employs simple methods involving random state preparation and measurement, without requiring ancillas or adaptive operations.
In addition, it can be efficiently simulated using a unitary 2-design, which can be implemented with shallow quantum circuits of depth $\mathcal{O}(\log\log\log d)$ composed of random Clifford gates~\cite{schuster2024random}.
These observations show that the optimal query complexity can be achieved by an algorithm with a simple structure.

The constant query complexity of Algorithm~\ref{alg:incoh} in the average-case scenario stems from a structural property of Haar-random unitaries.
The eigenvalues of the CUE can be modeled as interacting Brownian particles on a unit circle with inter-particle repulsion~\cite{dyson1962brownian}.
Thus, for $\epcue$, these eigenvalues behave as repulsive particles confined within an arc of length $2\sin^{-1}(\varep/2)$.
Consequently, the eigenangles from $\epcue$ are well-distributed within this region with high probability, leading to an eigenangle variance of order $\varep^2$.
In contrast, worst-case channels from $E_{\varep}$ channels have highly concentrated eigenangles;
Only one eigenangle differs significantly, resulting in an exponentially smaller eigenangle variance of order $\varep^2/d$.
Our proof of Theorem~\ref{thm:average} leverages this observation, showing that Algorithm \ref{alg:incoh} can certify channels having \textit{well-distributed eigenangles} with $\mathcal{O}(1/\varep^2)$ queries.
A detailed proof is given in Appendix \ref{appsec:average}.

We numerically simulate Algorithm \ref{alg:incoh} on unitary channels sampled from $\epcue$ and verify our theoretical results.
To sample unitary channels from $\epcue$, we apply the rejection sampling method using the eigenvalue distribution of the 2-Jacobi ensemble~\cite{dumitriua2002matrix}.
Then, for each sampled channel, we simulate the corresponding error probability, as shown in Fig.~\ref{fig:error}.
The average behavior of the error probability curves is independent of the dimension $d$, even for a low dimension such as $d=4$. 
Additionally, the variance in error probability greatly decreases as the dimension $d$ increases. 
This aligns with our theoretical prediction that the proportion of exceptional edge cases decays exponentially as $\exp(-\Omega(d))$. 
The figure also indicates that the required query complexity lies within realistic experimental ranges. 
Algorithm \ref{alg:incoh} requires approximately $10^5$ queries to certify unitary channels up to precision $\varepsilon = 0.01$, corresponding to a deviation of roughly $1\%$ in the worst-case basis.
This query count is comparable to the number of circuit executions reported in recent large-scale experiments, such as Google’s Willow processor, which performed up to $10^6$ surface-code cycles with a 1.1~$\mu$s repetition time~\cite{acharya2024quantum}.

We distinguish our result from those in Refs.~\cite{montanaro2013survey, rosenthal2024average}, which show that certification under \textit{average-case distance} requires a constant query complexity of $\mathcal{O}(1/\varep^2)$.
In our case, we show that the same query complexity suffices for certifying \textit{average-case channels} under the more stringent diamond distance.
Our approach is operationally meaningful, as certification under the diamond distance provides uniform performance guarantees across all input states, whereas certification under average-case distance ensures correctness only on a specific input state~\cite{wilde2013quantum}.
Accordingly, our result indicates that fully reliable certification is available for almost every unitary channel, offering a stronger and more practical contribution to reliable quantum information processing.

\section{Discussion}
\label{sec:discussion}

In this work, we have investigated the query complexity for unitary channel certification.
We proved that an exponential number of queries is required to certify all unitary channels, while coherent algorithms can achieve a quadratic speedup over incoherent algorithms.
We then proved that exponential hardness can be significantly relaxed for average-case unitary channels, which can be certified with a constant number of queries.

We highlight a notable technical contribution from our proof of Theorem \ref{thm:incoh}.
In many quantum hypothesis testing problems, proofs establishing query lower bounds for incoherent algorithms use a common technique: reducing the problem to distinguishing between a target object and an ensemble of slightly perturbed target objects~\cite{chen2022quantum, chen2022exponential, chen2022tight, chen2023efficient, fawzi2023incoherent, chen2024tight, rosenthal2024average, oh2024entanglement, liu2025quantum}.
Due to technical challenges, previous works relied on ensembles containing mixedness, such as an ensemble of mixed states or noisy channels.
Our proof overcomes this limitation by extending the technique to an ensemble consisting solely of unitary channels (see Appendix \ref{appsec:incoh} for details).
Thus, we anticipate further applications of our approach in future work, including potential extensions of this lower bound to continuous variable systems, where analogous certification challenges remain largely unexplored.

We suggest some intriguing directions for future research. 
Extending our average-case result to general CPTP channel would be a critical step for efficient certification in practice.
In this case, defining an appropriate measure of average-case CPTP channel would be essential.
One could also investigate the query-optimal coherent certification algorithm that does not rely on the inverse channel $\mathcal{E}_{U^\dagger}$.

\vspace{2em}

{\renewcommand\addcontentsline[3]{}\acknowledgments}
S.J. and C.O. were supported by the National Research Foundation of Korea Grants (No. RS-2024-00431768 and No. RS-2025-00515456) funded by the Korean government (Ministry of Science and ICT~(MSIT)) and the Institute of Information \& Communications Technology Planning \& Evaluation (IITP) Grants funded by the Korea government (MSIT) (No. IITP-2025-RS-2025-02283189 and IITP-2025-RS-2025-02263264).

\bibliography{references}

\clearpage
\begingroup
\onecolumngrid

\vspace{4em}
\appendix

\noindent
\textbf{ \LARGE{}Appendices }
\vspace{2em}


\section*{Contents}
\begin{itemize}
  \bl{A.} \hyperref[appsec:incoh]{Worst-case query complexity for incoherent algorithms} \hfill \pageref*{appsec:incoh}
  \bl{B.} \hyperref[appsec:coh]{Worst-case query complexity for coherent algorithms} \hfill \pageref*{appsec:coh}
    \begin{itemize}
      \bl{1.} \hyperref[appsubsec:coh_lower]{Lower bound} \hfill \pageref*{appsubsec:coh_lower}
      \bl{2.} \hyperref[appsubsec:coh_upper]{Upper bound} \hfill \pageref*{appsubsec:coh_upper}
    \end{itemize} 
  \bl{C.} \hyperref[appsec:average]{Average-case query complexity} \hfill \pageref*{appsec:average}
  \bl{D.} \hyperref[appsec:lemma]{Proof of technical lemmas} \hfill \pageref*{appsec:lemma}
    \begin{itemize}
      \bl{1.} \hyperref[appsubsec:F1F2]{Proof of Lemma 2} \hfill \pageref*{appsubsec:F1F2}
      \bl{2.} \hyperref[appsubsec:F3F4]{Proof of Lemma 3} \hfill \pageref*{appsubsec:F3F4}
      \bl{3.} \hyperref[appsubsec:qsvtpoly]{Proof of Lemma 5} \hfill \pageref*{appsubsec:qsvtpoly}
      \bl{4.} \hyperref[appsubsec:qsvtdelta]{Proof of Lemma 6} \hfill \pageref*{appsubsec:qsvtdelta}
      \bl{5.} \hyperref[appsubsec:anti-concentration]{Proof of Lemma 7} \hfill \pageref*{appsubsec:anti-concentration}
      \bl{6.} \hyperref[appsubsec:haar]{Proof of technical lemmas on Haar randomness} \hfill \pageref*{appsubsec:haar}
        \begin{itemize}
          \bl{a.} \hyperref[appsubsubsec:X_1moment]{Proof of Lemma 8} \hfill \pageref*{appsubsubsec:X_1moment}
          \bl{b.} \hyperref[appsubsubsec:X_2moment]{Proof of Lemma 9} \hfill \pageref*{appsubsubsec:X_2moment}
        \end{itemize}
    \end{itemize}
\end{itemize}

\setcounter{theorem}{0}

\vspace{1em}

\section{Worst-case query complexity for incoherent algorithms}
\label{appsec:incoh}

We consider unitary channel certification with incoherent algorithms.
We derive the query lower bound for an adaptive incoherent certification algorithm with an arbitrarily large ancillary system.

\begin{theorem}
    Consider an adaptive, incoherent algorithm with an arbitrarily large ancillary system, which tests whether $D(\mathcal{E}_U,\mathcal{E}_I) \geq \varep$ or $\mathcal{E}_U = \mathcal{E}_I$ with success probability at least $2/3$. For $\varep < 1/2$ and $d > 50\varep^2$, the required number of queries to $\mathcal{E}_U$ (or $\mathcal{E}_{U^\dagger}$) is $N = \Omega(d/\varepsilon^2)$.
\end{theorem}

\begin{proof}
    We first introduce a hypothesis test, which is a restricted version of the original certification.
    We then show that testing the hypothesis with success probability at least 2/3 requires $N=\Omega(d/\varep^2)$ queries.
    This implies that the same lower bound applies to the original certification, thereby completing the proof.

    We consider a hypothesis test to determine whether $\mathcal{E}_U$ is an identity channel or if it is sampled from an ensemble of $\varep$-perturbed unitary channels, $E_\varep$.
    Let the hypotheses
    \begin{align}
        H_0:\mathcal{E}_U=\mathcal{E}_I\quad \text{v.s.}\quad H_1:\mathcal{E}_U\sim E_\varep
    \end{align}
    given with equal probability.
    Since a channel $\mathcal{E}_U$ sampled from $E_\varep$ always satisfies $D(\mathcal{E}_U, \mathcal{E}_I)\geq\varep$ by construction, any certification algorithm can distinguish these two hypotheses.
    Thus, if the hypothesis test with probability at least 2/3 requires $N=\Omega(d/\varep^2)$ queries, the certification task with probability at least 2/3 also requires $N=\Omega(d/\varep^2)$ queries.
    Therefore, proving the lower bound of $N=\Omega(d/\varep^2)$ for the hypothesis test is sufficient to complete our proof.

    We choose the ensemble $E_\varep$ as
    \begin{align}
        E_\varep=\{\mathcal{E}_{U_\psi}:U_\psi\coleq I+(e^{is}-1)\ketbra{\psi}{\psi}\}_{\ket{\psi}},
    \end{align}
    where $\ket{\psi}$ is a $d$-dimensional Haar-random state and $s\coleq2\sin^{-1}(\varep/2)$.
    In what follows, we frequently omit the ket notation and simply write $\ket{\psi}$ as $\psi$ for notational simplicity.
    We now validate that $E_\varep$ is an $\varep$-perturbed unitary channel ensemble.
    The unitary operator $U_\psi$ shifts the phase by $s$ on the $\ket{\psi}$ basis, while behaving as an identity operator on all other orthogonal bases.
    Thus, the eigenangles of $U_\psi$ are all zero except for one $s$.
    Consequently, we have $\theta_{\min}=0$ and $\theta_{\max}=s$, where $[\theta_{\min},\theta_{\max}]$ is the shortest arc covering all eigenangles of $U_\psi$.
    From Lemma \ref{lemma:distance}, this is equivalent to $D(\mathcal{E}_{U_\psi},\mathcal{E}_I)=\varep$.
    This validates that $E_\varep$ is an $\varep$-perturbed unitary channel ensemble.

    \begin{figure}[ht]
        \centering
        \begin{overpic}[scale=0.4]{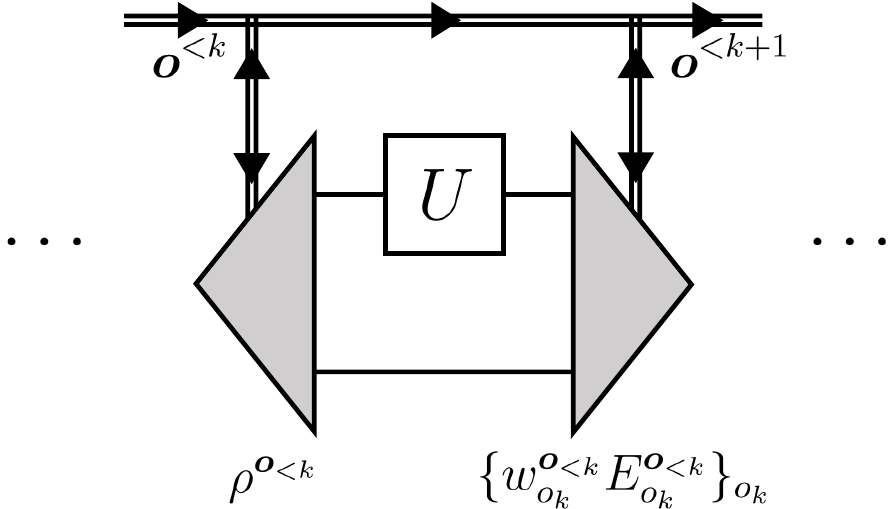} 
        \end{overpic}
        \caption{Schematic of an incoherent algorithm.}
        \label{fig:incoherent}
    \end{figure}

    Before proceeding to the proof, we clarify the setup and notations, providing a visualization in Fig.~\ref{fig:incoherent}.
    An incoherent algorithm performs state preparation and measurement for each of the $N$ queries to ancilla-coupled unitary channels $\mathcal{E}_U \otimes \mathcal{E}_{I_{\text{anc}}}$ or $\mathcal{E}_{U^\dagger} \otimes \mathcal{E}_{I_{\text{anc}}}$, where $\mathcal{E}_{I_{\text{anc}}}$ is the ancillary identity channel.
    Here, we restrict our attention to $\mathcal{E}_U$ without loss of generality, since the proof remains valid even if some instances of $\mathcal{E}_U$ are replaced with $\mathcal{E}_{U^\dagger}$.
    Let the testing algorithm yield measurement outcomes $\mathbf{o} = (o_1, \dots, o_N)$ from the $N$ queries.
    Here, the algorithm allows adaptivity, implying that the input state and measurement in the $k$-th query may depend on $\bo_{<k}\coleq(o_1,\dots,o_{k-1})$ for $2\leq k\leq N$.
    Accordingly, we denote the input state in the $k$-th query as $\rho^{\bo_{<k}}$.
    For the measurement, we introduce a POVM given by a set of $dd_{\text{anc}}$-dimensional operators $\{w_{o_{k}}^{{\bo_{<k}}}E_{o_{k}}^{{\bo_{<k}}}\}_{o_k}$, where $w_{o_{k}}^{{\bo_{<k}}}\geq 0$ and each $E_{o_{k}}^{{\bo_{<k}}}$ is positive semi-definite with $\tr(E_{o_{k}}^{{\bo_{<k}}})=1$.
    The POVM satisfies the completeness condition: $\sum _{o_k} w_{o_{k}}^{{\bo_{<k}}}E_{o_{k}}^{{\bo_{<k}}}=I\otimes I_{\text{anc}}$.
    Thus, the output probability distribution $p$ is given by
    \begin{align}
        p(o_{k}|\bo_{<k})=w_{o_{k}}^{{\bo_{<k}}}\tr (E_{o_{k}}^{{\bo_{<k}}}(\mathcal{E}_U\otimes\mathcal{E}_{I_{\text{anc}}})(\rho^{\bo_{<k}})).
    \end{align}
    This distribution $p$ depends on the underlying hypothesis and the choice of the Haar-random state $\ket{\psi}$.
    Thus, we denote $\mathbf{o} \sim p_0$ under hypothesis $H_0$, while $\mathbf{o} \sim p_{1,\psi}$ under hypothesis $H_1$ with $\mathcal{E}_U = \mathcal{E}_{U_\psi}$.
    Once all the measurements are completed, the algorithm chooses the correct hypothesis based on the measurement outcome $\mathbf{o}$.

    We now show that testing the hypothesis with success probability at least $2/3$ requires $N = \Omega(d/\varep^2)$ queries. 
    We establish this by showing that the success probability $P_{\text{success}}$ must be less than $2/3$ if $N \leq Cd/\varep^2$, where $C > 0$ is a constant.
    Our proof proceeds as follows:
    First, we show that it suffices to establish that the total variation distance (TVD) between the probability distributions $p_0$ and $\mathbb{E}_\psi p_{1,\psi}$ is less than $1/3$ for $N \leq Cd/\varep^2$.
    Next, we derive an upper bound on the TVD by partitioning the sample space of $(\psi,\bo)$ using a `good set' $G_\psi$. 
    The good set $G_\psi$ is defined as the set of outcomes $\bo$ for which the probabilities under both two hypotheses are similar, \textit{i.e.} $p_0(\bo) \approx p_{1,\psi}(\bo)$, and which satisfies certain additional technical conditions.
    Indeed, the random variable $(\psi,\bo)$ falls into one of two cases: either $\bo\notin G_\psi$ or $\bo\in G_\psi$.
    We separately bound the contributions from these two cases to the TVD.
    Specifically, we show that the random variable $(\psi, \bo)$ rarely falls into the first case $\bo\notin G_\psi$, and thus provide an upper bound on its contribution to the TVD. 
    We then derive a corresponding upper bound for the second case $\bo\in G_\psi$ using the property $p_0(\bo) \approx p_{1,\psi}(\bo)$, which again bounds the contribution for this case.
    Finally, by combining these upper bounds, we obtain an overall upper bound on the TVD as a function of $N$, thereby completing the proof.

    We show that it is sufficient to show that the TVD between the probability distributions $p_0$ and $\mathbb{E}_\psi p_{1,\psi}$ is less than $1/3$ for $N \leq Cd/\varep^2$.
    Here, TVD is a distance metric between probability distributions, defined as
    \begin{align}
        \text{TVD}(p,q)\coleq\frac{1}{2}\sum_{\bo}|p(\bo)-q(\bo)|
    \end{align}
    for distributions $p$ and $q$.
    To connect this with the success probability, we employ LeCam's two-point method~\cite{lecam1973convergence}, which relates the success probability to the TVD as follows:
    \begin{align}
        P_{\text{success}}
        &\leq1-\sum_{\bo}\min(p_0(\bo), \E_\psi p_{1,\psi}(\bo))\\
        &=\frac{1}{2}+\frac{1}{2}\text{TVD}(p_0,\E_\psi p_{1,\psi}).
    \end{align}
    Consequently, the upper bound $\text{TVD}(p_0,\E_\psi p_{1,\psi})<1/3$ implies that $P_{\text{success}}<2/3$.
    Therefore, we aim to show that $\text{TVD}(p_0,\E_\psi p_{1,\psi})<1/3$ holds for $N \leq Cd/\varep^2$.

    To this end, we partition the sample space of $(\psi,\bo)$ using a good set $G_\psi$, where $p_0(\mathbf{o}) \approx p_{1,\psi}(\mathbf{o})$ holds.
    The precise definition of this good set will be provided shortly.
    We first bound the TVD as follows:
    \begin{align}
        \text{TVD}(p_0,\E_\psi p_{1,\psi})
        &\equiv \frac{1}{2}\sum_{\bo}|p_0(\bo)-\E_\psi p_{1,\psi}(\bo)|\\
        &=\sum_{\bo}\max(0,p_0(\bo)-\E_\psi p_{1,\psi}(\bo))\\
        &\leq\sum_{\bo}\E_\psi\max(0,p_0(\bo)-p_{1,\psi}(\bo))\\
        &=\E_\psi\E_{\bo\sim p_0}\max\left(0,1-\frac{p_{1,\psi}(\bo)}{p_0(\bo)}\right)\\
        &=\E_\psi\E_{\bo\sim p_0}\max\left(0,1-L(\psi,\bo)\right),\label{eq:TVD_bound}
    \end{align}
    where $L(\psi,\bo)\coleq p_{1,\psi}(\bo)/p_{0}(\bo)$ is a likelihood ratio, and the third line follows from the convexity of the function $\max(0,\cdot)$.
    By partitioning the sample space using the good set, we derive an upper bound for the RHS as follows:
    \begin{align}
        &\E_\psi\E_{\bo\sim p_0}\max\left(0,1-L(\psi,\bo)\right)\\
        &\quad=\E_\psi\E_{\bo\sim p_0}\max(0, 1-L(\psi,\bo))(\ind_{\bo}((G_\psi)^c)+\ind_{\bo}(G_\psi))\\
        &\quad=\E_\psi\pr_{\bo\sim p_0}((G_\psi)^c)\E_{\bo\sim p_0|(G_\psi)^c}\max(0, 1- L(\psi,\bo))+\E_\psi\pr_{\bo\sim p_0}(G_\psi)\E_{\bo\sim p_0|G_\psi}\max(0, 1- L(\psi,\bo))\\
        &\quad\leq\E_\psi\pr_{\bo\sim p_0}((G_\psi)^c)+\E_\psi\E_{\bo\sim p_0|G_\psi}\max(0, 1- L(\psi,\bo)),\label{eq:TVD_bound2}
    \end{align}
    where the fourth line follows from $\max(0, 1-L(\psi,\bo))\leq 1$ and $\pr_{\bo\sim p_0}(G_\psi)\leq1$.
    Here, we introduced the indicator function notation:
    \begin{align}
        \ind_X(\text{condition of }X)\coleq
        \begin{cases} 
          1 & {X\text{ satisfies the condition}} \\
          0 & {X\text{ does not satisfy the condition}} 
       \end{cases}.
    \end{align}
    We also introduced the conditional expectation notation $\E_{\bo\sim p|A}$ for a distribution $p$ and a set $A$, which denotes expectation with respect to the distribution $p$ conditioned on the event $\bo\in A$.
    More precisely, for a function $f(\bo)$, we have
    \begin{align}
        \E_{\bo\sim p|A}f(\bo)=\frac{\sum_{\bo\in A}p(\bo)f(\bo)}{\sum_{\bo\in A}p(\bo)}.
    \end{align}
    Now, our goal is reduced to showing that the sum of the two terms in Eq.~\eqref{eq:TVD_bound2} is less than $1/3$ for $N \leq Cd/\varep^2$.
    These two terms correspond to two cases of the sample space: $\bo\notin G_\psi$ and $\bo\in G_\psi$.

    We now define the good set $G_\psi$ precisely. 
    To this end, let us rewrite $L(\psi,\bo)$ in a more convenient form.
    Since the algorithm is adaptive, the probability corresponding to $\bo$ can be expressed as
    \begin{align}
        p_{0}(\bo)
        &=p_0(o_1)\dots p_0(o_N|\bo_{<N}),\\
        p_{1,\psi}(\bo)
        &=p_{1,\psi}(o_1)\dots p_{1,\psi}(o_N|\bo_{<N})
    \end{align}
    under $H_0$, and $H_1$ with $\mathcal{E}_U=\mathcal{E}_{U_\psi}$, respectively.
    Given the input state $\rho^{\bo_{<k}}$ in $k$-th query, the conditional output probabilities are:
    \begin{align}
        p_0(o_k|\bo_{<k})&=w_{o_{k}}^{{\bo_{<k}}}\tr \left(E_{o_{k}}^{{\bo_{<k}}}\rho^{\bo_{<k}}\right),\\
        p_{1,\psi}(o_k|\bo_{<k})&=w_{o_{k}}^{{\bo_{<k}}}\tr \left(E_{o_{k}}^{{\bo_{<k}}}(U_\psi\otimes I_{\text{anc}})\rho^{\bo_{<k}}(U_\psi\otimes I_{\text{anc}})^{\dagger}\right).
    \end{align}
    Thus, we obtain
    \begin{align}
        L(\psi,\bo)
        &=\frac{p_{1,\psi}(\bo)}{p_0(\bo)}\\
        &=\prod_{k=1}^N\frac{p_{1,\psi}(o_k|\bo_{<k})}{p_0(o_k|\bo_{<k})}\\
        &=\prod_{k=1}^N\frac{\tr \left(E_{o_{k}}^{{\bo_{<k}}}(U_\psi\otimes I_{\text{anc}})\rho^{\bo_{<k}}(U_\psi\otimes I_{\text{anc}})^{\dagger}\right)}{\tr (E_{o_{k}}^{\bo_{<k}}\rho^{\bo_{<k}})}\\
        &=\prod_{k=1}^N\left(1+X_k({\psi,\bo})\right)\label{eq:def_L}
    \end{align}
    with $X_k({\psi,\bo})\coleq \tr \left(E_{o_{k}}^{{\bo_{<k}}}(U_\psi\otimes I_{\text{anc}})\rho^{\bo_{<k}}(U_\psi\otimes I_{\text{anc}})^{\dagger}\right)/\tr (E_{o_{k}}^{\bo_{<k}}\rho^{\bo_{<k}})-1$.
    Now, define the function:
    \begin{align}
        f(E,\rho)\coleq\frac{\tr(\tr_S(E)\tr_S(\rho))}{\tr(E\rho)},
    \end{align}
    where $\tr_S$ denotes the partial trace over the system Hilbert space $\mathcal{H}_S$.
    We are now ready to define the good set $G_\psi$.
    First, we introduce
    \begin{align}
        A_\alpha&\coleq\left\{\bo:\sum_{k=1}^N f(E_{o_k}^{\bo_{<k}},\rho^{\bo_{<k}})\leq\alpha N\right\},\\
        B_{\beta,\psi}&\coleq\{\bo:X_k(\psi,\bo)\geq-\beta\text{ for all }1\leq k\leq N\},\\
        C_{\gamma,\psi}&\coleq\left\{\bo:\sum_{k=1}^N \E_{\bo\sim p_0|\bo_{<k}}X_k^2(\psi,\bo)\leq \gamma\right\},
    \end{align}
    with parameters $\alpha$, $\beta$, and $\gamma$ satisfying
    \begin{align}
        \alpha&=100d,\label{eq:alpha_condition}\\
        \beta&>\frac{4s^2}{d}. \label{eq:beta_condition}
    \end{align}
    Then, we define the good set as the intersection 
    \begin{align}
        G_\psi\coleq A_\alpha\cap B_{\beta,\psi}\cap C_{\gamma,\psi}.
    \end{align}

    We briefly justify how this definition ensures the desired `good' properties, including similar probabilities $p_0(\bo)\approx p_{1,\psi}(\bo)$ for $\bo\in G_\psi$.
    The set $C_{\gamma,\psi}$ bounds the second moment of $X_k(\psi,\bo)$ over the $(\psi,\bo)$-space.
    This ensures concentration of $X_k(\psi,\bo)$ around zero, thus implying from Eq.~\eqref{eq:def_L} that $p_0(\bo)\approx p_{1,\psi}(\bo)$.
    However, Eq.~\eqref{eq:def_L} also indicates that even a single exception event, such as $X_k(\psi,\bo)\approx -1$, causes a substantial deviation between $p_0(\bo)$ and $p_{1,\psi}(\bo)$.
    The set $B_{\beta,\psi}$ with $\beta<1$ prevents such exceptions.
    The set $A_\alpha$ and conditions in Eqs.~\eqref{eq:alpha_condition} and \eqref{eq:beta_condition} are introduced for technical reasons.
    The exact values of the parameters $\beta$ and $\gamma$ will be specified later in the proof, and condition Eq.~\eqref{eq:beta_condition} will subsequently be validated.

    We now derive the upper bound on the RHS of Eq.~\eqref{eq:TVD_bound2}.
    We first present an upper bound on the first term:
    \begin{lemma}
        \label{lemma:F1F2}
        Let
        \begin{align*}
            g(s,d,N)&\coleq\frac{606s^2N}{d}+\frac{720072 s^4N^2}{d^2},\\
            F_1&\coleq0.01+\frac{20g(s,d,N)}{\beta^2},\\
            F_2&\coleq0.01+\frac{g(s,d,N)}{\gamma}.
        \end{align*}
        For $s<1$, $\alpha=100d$, and $\beta>4s^2/d$, we have
        \begin{align*}
            \E_\psi\pr_{\bo\sim p_0}((G_\psi)^c)\leq F_1+F_2.
        \end{align*}
    \end{lemma}
    \begin{proof}
        Appendix \ref{appsubsec:F1F2}
    \end{proof}
    \noindent We also derive the following upper bound for the second term:
    \begin{lemma}
        \label{lemma:F3F4}
        Let
        \begin{align*}
            F_3&\coleq1-\exp\left(-\left(1+\frac{1}{\beta}\right)\gamma-\eta\right),\\
            F_4&\coleq\exp\left(-\frac{\eta^2}{4\gamma+2\beta\eta/3}\right).
        \end{align*}
        For $\eta,\gamma>0$, we have
        \begin{align*}
            \E_\psi\E_{\bo\sim p_0|G_\psi}\max(0, 1- L(\psi,\bo))\leq F_3+F_4.
        \end{align*}
    \end{lemma}
    \begin{proof}
        Appendix \ref{appsubsec:F3F4}
    \end{proof}
    \noindent Thus, combining the above results, we have
    \begin{align}
        \text{TVD}(p_0,\E_\psi p_{1,\psi})\leq F_1+F_2+F_3+F_4.
    \end{align}
    Now we choose the parameters explicitly as follows:
    \begin{align}
        \beta&=0.1,\\
        \gamma&=0.0003,\\
        \eta&=0.3.
    \end{align}
    We validate the assumption $\beta>4s^2/d$ in Eq.~\eqref{eq:beta_condition}.
    Recall that $s=2\sin^{-1}(\varep/2)$ is a function of $\varep$.
    For $\varep=1/2$, we have $s=2\sin^{-1}(1/4)=0.505\dots<1.02/2=1.02\varep$.
    Since $s$ is convex on $\varep\in[0,1/2)$, it holds that $0\leq s<1.02\varep$ for any $\varep<1/2$.  
    Thus, given the assumption $d>50\varep^2$, we have
    \begin{align}
        \beta=0.1>\frac{5\varep^2}{d}>\frac{5s^2}{1.02^2d}>\frac{4s^2}{d}.
    \end{align}
    Hence, our choice of $\beta$ satisfies Eq.~\eqref{eq:beta_condition}.
    Under these parameters, we obtain
    \begin{align}
        F_1&=0.01+2\times10^3g(s,d,N),\\
        F_2&=0.01+3.333\ldots\times10^3g(s,d,N),\\
        F_3&=0.261\dots,\\
        F_4&=0.014\dots.
    \end{align}
    Therefore, the TVD is bounded above by
    \begin{align}
        \text{TVD}(p_0,\E_\psi(p_{1,\psi}))
        &\leq 0.295\dots+5.333\ldots\times10^3 g(s,d,N).
    \end{align}
    When $N\leq10^{-8}d/s^2$, we have
    \begin{align}
        g(s,d,N)
        &\equiv\frac{606s^2N}{d}+\frac{720072 s^4N^2}{d^2}\\
        &\leq606\times10^{-8}+720072\times10^{-16}\\
        &=6.060\dots\times10^{-6}.
    \end{align}
    Thus, the upper bound on the TVD becomes
    \begin{align}
        \text{TVD}(p_0,\E_\psi(p_{1,\psi}))\leq0.328\dots<\frac{1}{3}.
    \end{align}
    Consequently, $N>10^{-8}d/s^2$ queries are necessary for the TVD to exceed $1/3$, corresponding to a success probability greater than $2/3$.
    This result implies that certification with success probability at least 2/3 requires query complexity $N=\Omega(d/s^2)=\Omega(d/\varep^2)$, completing the proof.
\end{proof}

\section{Worst-case query complexity for coherent algorithms}
\label{appsec:coh}

We consider unitary channel certification with coherent algorithm.
We derive the tight query complexity of a coherent certification algorithm with an arbitrarily large ancillary system.

\subsection{Lower bound}
\label{appsubsec:coh_lower}

We derive a query lower bound for the incoherent certification algorithm.

\begin{theorem}
    Consider a coherent algorithm with an arbitrarily large ancillary system, which tests whether $D(\mathcal{E}_U,\mathcal{E}_I)\geq \varep$ or $\mathcal{E}_U=\mathcal{E}_I$ with success probability at least $2/3$. For $\varep<1/2$, the required number of queries to $\mathcal{E}_U$ (or $\mathcal{E}_{U^\dagger}$) is $N=\Omega({\sqrt{d}}/{\varep})$.
\end{theorem}

\begin{proof}
    As in the case of an incoherent algorithm, we consider a restricted hypothesis testing problem, which is an easier task than the original certification problem.
    As well, we restrict our attention to $\mathcal{E}_U$ without loss of generality, since the proof remains valid even if some instances of $\mathcal{E}_U$ are replaced with $\mathcal{E}_{U^\dagger}$.
    Suppose the given channel under hypothesis $H_0$ is an identity channel denoted by $\mathcal{E}_U=\mathcal{E}_I$, and under hypothesis $H_1$, is sampled from an ensemble of $\varep$-perturbed unitary channels $E_\varep$, where $E_\varep$ is defined as follows: 
    \begin{align}
        E_\varep=\{\mathcal{E}_{U_\psi}:U_\psi\coleq I+(e^{is}-1)\ketbra{\psi}{\psi}\}_{\ket{\psi}}.
    \end{align}
    Here, $\ket{\psi}$ is a Haar-random state, $s\coleq2\sin^{-1}(\varep/2)$, and the hypotheses are given with equal probability.
    It is sufficient to show that the hypothesis testing with success probability at least 2/3 requires $\Omega(\sqrt{d}/\varep)$ queries.
    Thus, we need to show that the success probability $P_{\text{success}}$ is less than 2/3 for $N\leq C\sqrt{d}/\varep$ with a constant $C>0$.

    \begin{figure}[ht]
        \centering
        \begin{overpic}[scale=0.4]{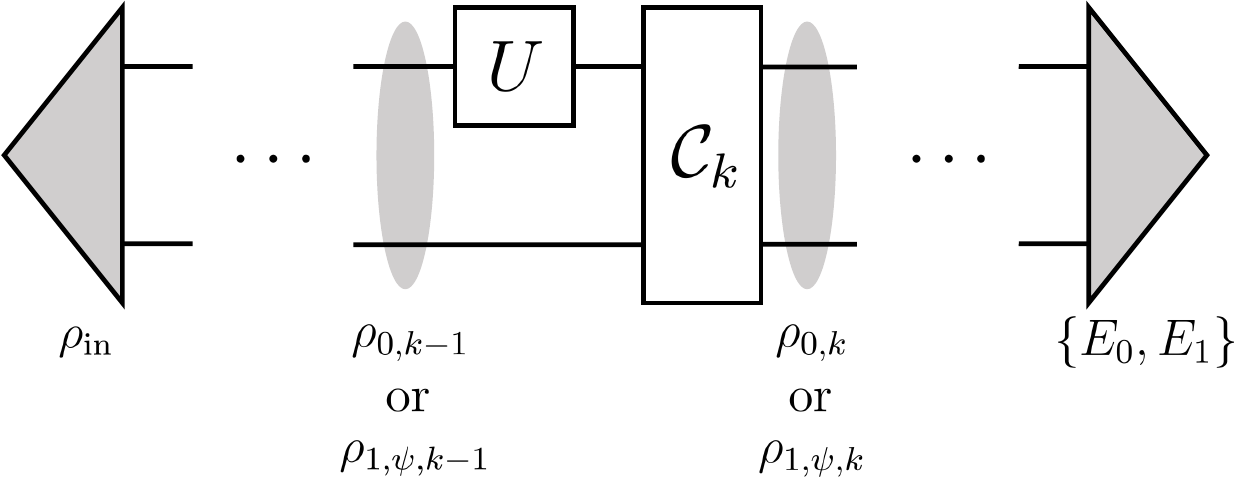} 
        \end{overpic}
        \caption{Schematic of a coherent algorithm.}
        \label{fig:coherent}
    \end{figure}

    We first show an upper bound on the success probability.
    Let
    \begin{align}
        \rho_{0,k}&\coleq (\mathcal{E}_I\otimes\mathcal{E}_{I_{\text{anc}}})(\mathcal{C}_{k}(\rho_{0,k-1})),\\
        \tilde{\rho}_{1,\psi,k}&\coleq (\mathcal{E}_{U_\psi}\otimes\mathcal{E}_{I_{\text{anc}}})(\mathcal{C}_{k}(\rho_{0,k-1})),\\
        \rho_{1,\psi,k}&\coleq (\mathcal{E}_{U_\psi}\otimes\mathcal{E}_{I_{\text{anc}}})(\mathcal{C}_{k}(\rho_{1,\psi,k-1})),
    \end{align}
    where $\mathcal{C}_k$ is a CPTP map corresponding to the quantum algorithm at the $k$-th step, and $\rho_{0,0}=\rho_{1,\psi,0}=\rho_{\text{in}}$, as shown in Fig.~\ref{fig:coherent}. 
    The output state $\rho_{\text{out}}$ is then given either by $\rho_{\text{out}}=\rho_{0,N}$ or $\rho_{\text{out}}=\rho_{1,\psi,N}$ depending on the hypothesis.
    Let the algorithm measure the output state with a POVM $\{E_0, E_1\}$ satisfying $E_0+E_1=I\otimes I_{\text{anc}}$, where $E_0$ and $E_1$ correspond to $H_0$ and $H_1$ respectively, without loss of generality.
    Then, the upper bound of the success probability is given as
    \begin{align}
        P_{\text{success}}
        &\leq\frac{1}{2}\max_{E_0,E_1}\E_\psi(\tr(E_0\rho_{0,N})+\tr(E_1\rho_{1,\psi,N}))\\
        &\leq\frac{1}{2}\E_\psi\max_{E_0,E_1}(\tr(E_0\rho_{0,N})+\tr(E_1\rho_{1,\psi,N}))\\
        &\leq\frac{1}{2}+\frac{1}{4}\E_\psi\|\rho_{0,N}-\rho_{1,\psi,N}\|_1,
    \end{align}
    where the last line follows from the Helstrom bound~\cite{helstrom1969quantum}.

    Thus, it is sufficient to show that $\E_\psi\|\rho_{0,N}-\rho_{1,\psi,N}\|_1<2/3$ if $N\leq C\sqrt{d}/\varep$. 
    Using the triangular inequality, we have
    \begin{align}
        \E_\psi\|\rho_{0,N}-\rho_{1,\psi,N}\|_1\leq\E_\psi\|\rho_{0,N}-\tilde{\rho}_{1,\psi,N}\|_1+\E_\psi\|\tilde{\rho}_{1,\psi,N}-\rho_{1,\psi,N}\|_1.\label{eq:distance_triangle}
    \end{align}
    We now derive the upper bound on each term on the RHS.
    We first consider the first term.
    Writing $\xi=\mathcal{C}_N(\rho_{0,N-1})$, we have
    \begin{align}
        \E_\psi\|\rho_{0,N}-\tilde{\rho}_{1,\psi,N}\|_1
        &=\E_\psi\|(\mathcal{E}_I\otimes\mathcal{E}_{I_{\text{anc}}})(\xi)-(\mathcal{E}_{U_\psi}\otimes\mathcal{E}_{I_{\text{anc}}})(\xi)\|_1\\
        &=\E_\psi\|\xi-(U_\psi\otimes I_{\text{anc}})\xi(U_\psi^\dagger\otimes I_{\text{anc}})\|_1\\
        &=\E_\psi\|(e^{i s}-1)(\ketbra{\psi}{\psi}\otimes I_{\text{anc}})\xi+(e^{-i s}-1)\xi(\ketbra{\psi}{\psi}\otimes I_{\text{anc}})\notag\\
        &\quad+(2-e^{i s}-e^{-i s})(\ketbra{\psi}{\psi}\otimes I_{\text{anc}})\xi(\ketbra{\psi}{\psi}\otimes I_{\text{anc}})\|_1\\
        &\leq(2|e^{is}-1|+|2-e^{i s}-e^{-i s}|)\E_\psi\|(\ketbra{\psi}{\psi}\otimes I_{\text{anc}})\xi\|_1\\
        &=(2\sqrt{2-2\cos s}+2-2\cos s)\E_\psi\|(\ketbra{\psi}{\psi}\otimes I_{\text{anc}})\xi\|_1\\
        &\leq3 s\E_\psi\|(\ketbra{\psi}{\psi}\otimes I_{\text{anc}})\xi\|_1
    \end{align}
    for $s<1$, where the fourth line follows from Hölder's inequality
    \begin{align}
        \|(\ketbra{\psi}{\psi}\otimes I_{\text{anc}})\xi(\ketbra{\psi}{\psi}\otimes I_{\text{anc}})\|_1
        &\leq\|(\ketbra{\psi}{\psi}\otimes I_{\text{anc}})\xi\|_1\|(\ketbra{\psi}{\psi}\otimes I_{\text{anc}})\|_\infty\\
        &=\|(\ketbra{\psi}{\psi}\otimes I_{\text{anc}})\xi\|_1.
    \end{align}
    From the eigendecomposition of $\xi=\sum_{k=1}^{dd_{\text{anc}}}\lambda_k\ketbra{\eta_k}{\eta_k}$, we obtain
    \begin{align}
        \E_\psi\|(\ketbra{\psi}{\psi}\otimes I_{\text{anc}})\xi\|_1
        &\leq\sum_{k=1}^{dd_{\text{anc}}}\lambda_k\E_\psi\|(\ketbra{\psi}{\psi}\otimes I_{\text{anc}})\ketbra{\eta_k}{\eta_k}\|_1\\
        &=\sum_{k=1}^{dd_{\text{anc}}}\lambda_k\E_\psi|\tr((\ketbra{\psi}{\psi}\otimes I_{\text{anc}})\ketbra{\eta_k}{\eta_k})|
            && (\ketbra{\psi}{\psi}\otimes I_{\text{anc}})\ketbra{\eta_k}{\eta_k}\text{ is rank-1}\\
        &=\sum_{k=1}^{dd_{\text{anc}}}\lambda_k\E_\psi\sqrt{\bra{\eta_k}(\ketbra{\psi}{\psi}\otimes I_{\text{anc}})\ket{\eta_k}^2}\\
        &\leq\sum_{k=1}^{dd_{\text{anc}}}\lambda_k\E_\psi\sqrt{\bra{\eta_k}(\ketbra{\psi}{\psi}\otimes I_{\text{anc}})\ket{\eta_k}}
            && \|\ketbra{\psi}{\psi}\otimes I_{\text{anc}}\|_\infty=1\\
        &\leq\sum_{k=1}^{dd_{\text{anc}}}\lambda_k\sqrt{\E_\psi\bra{\eta_k}(\ketbra{\psi}{\psi}\otimes I_{\text{anc}})\ket{\eta_k}}
            &&\text{Concavity of square root}\\
        &=\sum_{k=1}^{dd_{\text{anc}}}\lambda_k\sqrt{\bra{\eta_k}(I/d\otimes I_{\text{anc}})\ket{\eta_j}}\\
        &=\frac{1}{\sqrt{d}}.
    \end{align}
    Thus, we have obtained the upper bound on the first term as follows:
    \begin{align}
        \label{eq:coherent_1}
        \E_\psi\|\rho_{0,N}-\tilde{\rho}_{1,\psi,N}\|_1\leq\frac{3s}{\sqrt{d}}.
    \end{align}

    We now turn to the second term.
    Since a CPTP map cannot increase the trace distance between two states~\cite{nielsen2010quantum}, we obtain
    \begin{align}
        \E_\psi\|\tilde{\rho}_{1,\psi,N}-\rho_{1,\psi,N}\|_1
        &=\E_\psi\|(\mathcal{E}_{U_\psi}\otimes\mathcal{E}_I)(\mathcal{C}_{k}(\rho_{0,N-1}))-(\mathcal{E}_{U_\psi}\otimes\mathcal{E}_I)(\mathcal{C}_{k}(\rho_{1,\psi,N-1}))\|_1\\
        &\leq\E_\psi\|\rho_{0,N-1}-\rho_{1,\psi,N-1}\|_1\\
        &\leq\E_\psi\|\rho_{0,N-1}-\tilde{\rho}_{1,\psi,N-1}\|_1+\E_\psi\|\tilde{\rho}_{1,\psi,N-1}-\rho_{1,\psi,N-1}\|_1\\
        &\leq\sum_{k=1}^{N-1}\E_\psi\|\rho_{0,k}-\tilde{\rho}_{1,\psi,k}\|_1\\
        &\leq\frac{3s(N-1)}{\sqrt{d}},  \label{eq:coherent_2}
    \end{align}
    where the fourth line follows from the induction and the last line follows from Eq.~\eqref{eq:coherent_1}.
    Combining the two upper bounds in Eqs.~\eqref{eq:coherent_1} and \eqref{eq:coherent_2}, we obtain
    \begin{align}
        \E_\psi\|\rho_{0,N}-\rho_{1,\psi,N}\|_1
        &\leq\E_\psi\|\rho_{0,N}-\tilde{\rho}_{1,\psi,N}\|_1+\E_\psi\|\tilde{\rho}_{1,\psi,N}-\rho_{1,\psi,N}\|_1\\
        &\leq\frac{3sN}{\sqrt{d}}.
    \end{align}
    Hence, for $N\leq\sqrt{d}/6s$, we have
    \begin{align}
        \E_\psi\|\rho_{0,N}-\rho_{1,\psi,N}\|_1\leq\frac{1}{2}<\frac{2}{3}.
    \end{align}
    This implies that the query complexity $N=\Omega(\sqrt{d}/s)=\Omega(\sqrt{d}/\varep)$ is required, thereby completing the proof.
\end{proof}

\subsection{Upper bound}
\label{appsubsec:coh_upper}

We derive the query upper bound on the incoherent certification algorithm by proposing a query-optimal coherent algorithm based on QSVT.
\begin{theorem}
    There exists a coherent algorithm that tests whether $\mathcal{E}_U=\mathcal{E}_I$ or $D(\mathcal{E}_U, \mathcal{E}_I)\geq\varep$ with success probability at least $2/3$ using $N=\mathcal{O}(\sqrt{d}/\varep)$ queries to $\mathcal{E}_U$ and $\mathcal{E}_{U^\dagger}$.
\end{theorem}

\begin{proof}
    We prove that Algorithm \ref{alg:coh} achieves the optimal query complexity $N=\mathcal{O}(\sqrt{d}/\varep)$.
    Before proceeding, we briefly outline the algorithm.
    Our certification algorithm of the unitary channel $\mathcal{E}_U$ proceeds in three steps: first, prepare a Haar-random state $\ket{\psi}$; second, apply it to the QSVT circuit $\mathcal{E}_{V_\Phi}$; third, measure the output state with the POVM $\{{\ketbra{\psi}{\psi}, I-\ketbra{\psi}{\psi}}\}$.
    The measurement outcome $\ketbra{\psi}{\psi}$ implies the decision $H_0: \mathcal{E}_U=\mathcal{E}_I$; otherwise, we conclude $H_1: D(\mathcal{E}_U, \mathcal{E}_I)\geq\varepsilon$.

    We now analyze the query complexity of our algorithm. 
    First, we construct the QSVT operator $V_\Phi$ and derive an upper bound on the error probability in terms of the corresponding polynomial $P$.
    We then identify a polynomial $P$ that ensures an error probability of at most $1/3$. 
    Lastly, we show that the construction of $V_\Phi$ using this polynomial requires $\mathcal{O}(\sqrt{d}/\varepsilon)$ queries to $\mathcal{E}_U$.

    We construct the QSVT operator $V_\Phi$ with the orthogonal projections $\Pi_\psi\coleq\ketbra{\psi}{\psi}$ and $U\Pi_\psi U^\dagger$, along with the identity operator $V=I$, following the notation from Lemma \ref{lemma:qsvt}. 
    The block encoding of $V$ is given as
    \begin{align}
        V=\kbordermatrix{
        &\Pi_\psi &  \\
        U\Pi_\psi U^\dagger & \bra{\psi}U^\dagger\ket{\psi} & \cdot \\ 
        & \cdot & \cdot 
    }.
    \end{align}
    Consequently, the resulting QSVT operator $V_\Phi$ is block-encoded as
    \begin{align}
        V_\Phi=\kbordermatrix{
        &\Pi_\psi &  \\
        \Pi_\psi & P^{(\text{SV})}(\bra{\psi}U^\dagger\ket{\psi}) & \cdot \\ 
        & \cdot & \cdot 
    },
    \end{align}
    where the polynomial $P$ satisfies: (1) $\deg(P)$ is even, (2) $P$ is even, and (3) $|P(x)|\leq 1$ for $x\in[-1, 1]$.

    Next, we express the error probabilities in terms of $P$.
    Let $P_{\text{error}|H_0}$ and $P_{\text{error}|H_1}$ denote the error probabilities under hypotheses $H_0$ and $H_1$, respectively. 
    These probabilities are expectations over the Haar-random state $\ket{\psi}$, given by $P_{\text{error}|H_i}=\E_{\psi}P_{\text{error}|H_i, \psi}$ for $i\in{0,1}$.
    Under hypothesis $H_0$, the conditional error probability is
    \begin{align}
        P_{\text{error}|H_0,\psi}
        &=1-|\bra{\psi}V_\Phi\ket{\psi}|^2\\
        &=1-|P^{(\text{SV})}(\bra{\psi}U^\dagger\ket{\psi})|^2\\
        &=1-(P(|\bra{\psi}U^\dagger\ket{\psi}|))^2\\
        &=1-(P(1))^2,
    \end{align}
    where the second line follows from $|P^{(SV)}(c)|=P(|c|)$ for complex $c$, and the last line follows from the condition $\mathcal{E}_U=\mathcal{E}_I$ of $H_0$.
    Similarly, under $H_1$, we have
    \begin{align}
        P_{\text{error}|H_1,\psi}
        &=|\bra{\psi}V_\Phi\ket{\psi}|^2\\
        &=|P^{(\text{SV})}(\bra{\psi}U^\dagger\ket{\psi})|^2\\
        &=(P(|\bra{\psi}U^\dagger\ket{\psi}|))^2.
    \end{align}
    We impose an additional condition: (4) $P(1)=1$, which allows us to neglect the $H_0$ error.
    Then, defining the set $S_{\delta}\coleq\{\psi:|\bra{\psi}U^\dagger\ket{\psi}|\leq1-\delta\}$, we bound the error probability as follows:
    \begin{align}
        P_{\text{error}}
        &=P_{\text{error}|H_1}\\
        &=\E_\psi P_{\text{error}|H_1,\psi}\\
        &=\E_\psi (P(|\bra{\psi}U^\dagger\ket{\psi}|))^2\\
        &=\E_\psi (P(|\bra{\psi}U^\dagger\ket{\psi}|))^2\ind_\psi(S_{\delta}) + \E_\psi (P(|\bra{\psi}U^\dagger\ket{\psi}|))^2\ind_\psi((S_{\delta})^c)\\
        &\leq \E_{\psi|S_{\delta}} (P(|\bra{\psi}U^\dagger\ket{\psi}|))^2+\pr_\psi ((S_{\delta})^c)\\
        &\leq \max_{x\in[0, 1-\delta]}(P(x))^2+\pr_\psi ((S_{\delta})^c).
    \end{align}

    From this, we now show that the error probability is at most $1/3$ for some $\delta$ and $P$ satisfying all the imposed assumptions.
    We establish two lemmas to derive upper bounds on each term on the RHS.
    The first lemma explicitly provides a polynomial $P$ that ensures a small upper bound on the first term: 
    \begin{lemma}
        \label{lemma:qsvtpoly}
        For $0<\delta, \Delta<1/2$ and an $n$-th order Chebyshev polynomial $T_n$ with $n=2\ceil{1/\sqrt{\delta}\log(2/\Delta)}$, the polynomial
        \begin{align}
            P(x)=\frac{T_n(x/(1-\delta))}{T_n(1/(1-\delta))}
        \end{align}
        satisfies the following conditions:
        \begin{enumerate}[label=(\arabic*)]
            \item $\deg(P)$ is even.
            \item $P$ is even.
            \item $|P(x)|\leq 1$ for $x\in[-1, 1]$
            \item $P(1)=1$
            \item $|P(x)|\leq\Delta$ for $x\in[0,1-\delta]$.
        \end{enumerate}
    \end{lemma}
    \begin{proof}
        Appendix \ref{appsubsec:qsvtpoly}.
    \end{proof}
    \noindent
    An illustration of $P(x)$ is provided in Fig.~\ref{fig:poly}.
    \begin{figure}[ht]
        \centering
        \begin{overpic}[width=0.40\textwidth]{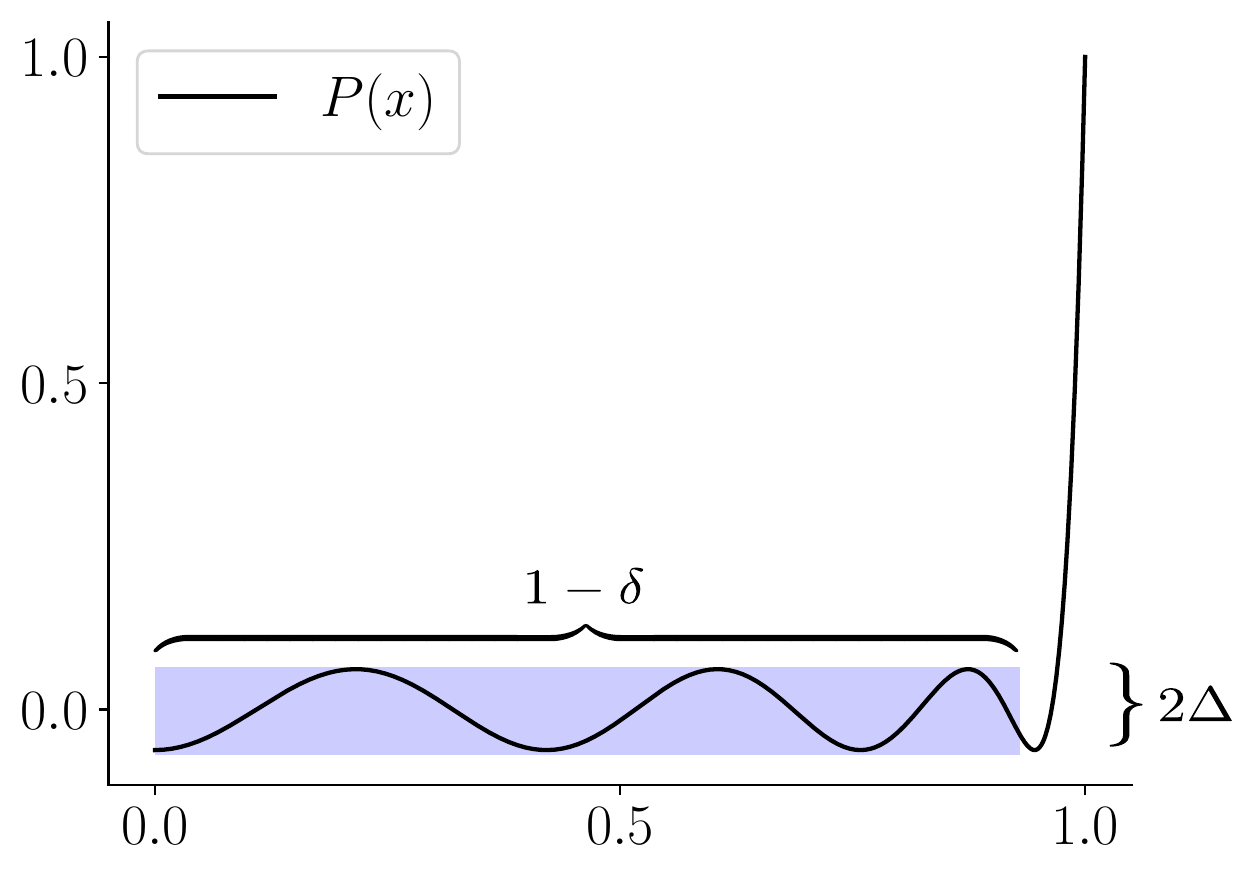} 
        \end{overpic}
        \caption{Illustration of polynomial $P(x)$ for Algorithm \ref{alg:coh}. The blue region indicates where the absolute value of the polynomial is bounded by $\Delta$.}
        \label{fig:poly}
    \end{figure}
    The fifth condition leads to the following upper bound for the first term:
    \begin{align}
        \max_{x\in[0, 1-\delta]}(P(x))^2\leq\Delta^2.
    \end{align}
    The second lemma ensures that for small $\delta$, the set $S_\delta$ has high probability, thus bounding the second term:
    \begin{lemma}
        \label{lemma:qsvtdelta}
        Let $S_{\delta}=\{\psi:|\bra{\psi}U^\dagger\ket{\psi}|\leq1-\delta\}$.
        For a Haar-random state $\ket{\psi}$ and unitary operator $U$ satisfying $D(\mathcal{E}_U, \mathcal{E}_I)\geq\varep$,
        \begin{align}
            \pr_\psi ((S_{\delta})^c)\leq\frac{8d\delta}{\varep^2}
        \end{align}
    \end{lemma}
    \begin{proof}
        Appendix \ref{appsubsec:qsvtdelta}.
    \end{proof}
    \noindent 
    Combining these lemmas, we have
    \begin{align}
        P_{\text{error}}\leq\Delta^2+\frac{8d\delta}{\varep^2}.
    \end{align}
    Thus, choosing $\Delta=1/\sqrt{6}$ and $\delta=\varep^2/(48d)$ yields the small error probability $P_{\text{error}}\leq1/3$. 
    This can be realized using a QSVT circuit with polynomial degree $2\ceil{(\sqrt{48}\log 2\sqrt{6})\sqrt{d}/\varep}$.

    The polynomial $P(x)$ has a degree of $\mathcal{O}(\sqrt{d}/\varep)$.
    Fig.~\ref{fig:qsvt_circuit} shows the implementation of the QSVT circuit $V_\Phi$ with polynomial degree $\mathcal{O}(\sqrt{d}/\varepsilon)$ using $\mathcal{O}(\sqrt{d}/\varepsilon)$ queries to $\mathcal{E}_U$ and $\mathcal{E}_{U^\dagger}$, therefore completing the proof.
    The QSVT circuit involves repeated rotations using $\Pi_\psi$ and $U\Pi_\psi U^\dagger$, as illustrated in Fig.~\ref{fig:qsvt_circuit} (a).
    The total required number of rotations $n$ is equal to the polynomial degree, implying $n=\mathcal{O}(\sqrt{d}/\varep)$ queries.
    Each $U\Pi_\psi U^\dagger$ rotation requires one query to both $\mathcal{E}_U$ and $\mathcal{E}_{U^\dagger}$, whereas $\Pi_\psi$ rotations require none. 
    Therefore, we need $n/2=\mathcal{O}(\sqrt{d}/\varepsilon)$ queries to each of $\mathcal{E}_U$ and $\mathcal{E}_{U^\dagger}$, which suffices for our result.

    \vspace{1.5em}
    \begin{figure}[ht]
        \centering
        \begin{overpic}[width=0.60\textwidth]{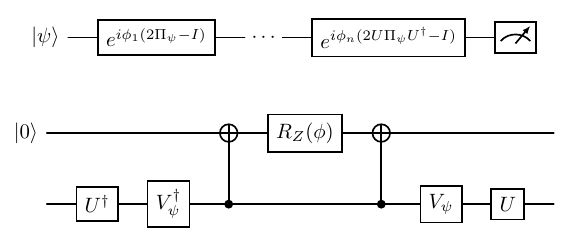}
            \put(-3,41.5){\text{(a)}} 
            \put(-3,26){\text{(b)}} 
        \end{overpic}
        \caption{Quantum circuits for Algorithm \ref{alg:coh}. (a) Full circuit. (b) Gate $e^{i\phi(2U\Pi_\phi U^\dagger-I)}$.}
        \label{fig:qsvt_circuit}
    \end{figure}

\end{proof}

\section{Average-case query complexity}
\label{appsec:average}
We show that almost all $\varep$-perturbed unitary channels can be certified using $\mathcal{O}(1/\varep^2)$ queries.
We first introduce necessary preliminaries, followed by the proof of the main theorem.

We begin by restating the average-case unitary channel ensemble $\epcue$ in terms of its probability density function (pdf).
By definition, $\epcue$ is the marginal distribution of CUE conditioned on an $\varep$-perturbation from the identity channel.
Therefore, the pdf of $\epcue$ can be expressed in terms of the pdf of the CUE.
Let $U_{\btheta}$ denote a unitary operator with eigenangles $\btheta=(\theta_1,\dots,\theta_d)$.
Then, the pdf of the eigenangles for an operator $U_{\btheta}\sim\text{CUE}$ is given by
\begin{align}
    f_\text{CUE}(\btheta)\coleq\frac{1}{C}\prod_{1\le k<l\le d}|e^{i\theta_k}-e^{i\theta_l}|^2,
\end{align}
where $C$ is a normalization constant~\cite{dyson1962algebraic}.
Consequently, the pdf for eigenangles of $U_{\btheta}\sim\epcue$ is
\begin{align}
    f_{\epcue}(\btheta)\coleq\frac{1}{C_\varep}\prod_{1\le k<l\le d}|e^{i\theta_k}-e^{i\theta_l}|^2\ind_{\btheta}(D(\mathcal{E}_{U_{\btheta}}, \mathcal{E}_I)=\varep)
\end{align}
with a normalization constant $C_\varep$.
Note that although the unitary operator $U_{\btheta}$ is not uniquely defined by $\btheta$, the indicator function $\ind_{\btheta}(D(\mathcal{E}_{U_{\btheta}}, \mathcal{E}_I))$ remains well-defined, since $D(\mathcal{E}_{U_{\btheta}}, \mathcal{E}_I)$ remains invariant under global phases of $U_{\btheta}$.
This invariance also allows us to reduce the condition $D(\mathcal{E}_{U_{\btheta}}, \mathcal{E}_I)=\varep$ to
\begin{align}
    \btheta\in\mathcal{R}(s)\coleq \{\btheta:\min\btheta=-s/2,\max\btheta=s/2\}
\end{align}
without loss of generality, recalling that $s=2\sin^{-1}(\varep/2)$ is determined by $\varep$.
Hence, we redefine the pdf $f_{\epcue}$ as
\begin{align}
    \label{eq:epCUE}
    f_{\epcue}(\btheta)\coleq\frac{1}{C_\varep'}\prod_{1\le k<l\le d}|e^{i\theta_k}-e^{i\theta_l}|^2\ind_{\btheta}({\mathcal{R}(s)})
\end{align}
with a normalization constant $C_\varep'$.
In this appendix, we use the notation $\btheta\sim\epcue$ to indicate that $\btheta$ follows the pdf $f_{\epcue}$, which is a slightly abused notation, as $\epcue$ originally denotes an ensemble of unitary operators.

Now, we prove the main theorem.
\begin{theorem}
    Suppose a random unitary channel $\mathcal{E}_U$ is given with $U\sim\epcue$ under $\varep<1/2$ and dimension $d\geq4$. 
    There exists an algorithm that tests whether $D(\mathcal{E}_U, \mathcal{E}_I)\geq\varep$ or $\mathcal{E}_U=\mathcal{E}_I$ with success probability at least 2/3 using $N=\mathcal{O}(1/\varep^2)$ queries, except for an $\exp(-\Omega(d))$ fraction of $U$.
\end{theorem}
\begin{proof}
    We show that Algorithm \ref{alg:incoh} suffices.
    Recall the three steps of the algorithm:
    \begin{enumerate}
        \item Prepare a Haar-random state $\ket{\psi}$ and measure the state $U\ket{\psi}$ with POVM $\{\ketbra{\psi}{\psi}, I-\ketbra{\psi}{\psi}\}$. Each measurement yields an outcome of 0 or 1, respectively.
        \item Repeat this procedure $N$ times. Let $X_k$ be the $k$-th measurement outcome and $\ket{\psi_k}$ be the corresponding $k$-th Haar-random state.
        \item If at least one measurement yields $X_k=1$, conclude $D(\mathcal{E}, \mathcal{E}_I)\geq\varep$. Otherwise, conclude $\mathcal{E}_U=\mathcal{E}_I$.
    \end{enumerate}
    Under the null hypothesis $H_0:\mathcal{E}_U = \mathcal{E}_I$, the algorithm never incorrectly outputs $D(\mathcal{E}_U, \mathcal{E}_I) \geq \varepsilon$.
    Therefore, the only possible error occurs under hypothesis $H_1$: incorrectly deciding $\mathcal{E}_U = \mathcal{E}_I$ when $D(\mathcal{E}_U, \mathcal{E}_I) \geq \varepsilon$.

    We now show that the error probability under hypothesis $H_1$ can be made less than $1/3$ with $\mathcal{O}(1/\varepsilon^2)$ queries to $\mathcal{E}_U$, valid for a $1 - \exp(-\Omega(d))$ fraction of unitary operators $U \sim \epcue$.
    The proof involves three steps:
    First, we express the $H_1$ error probability in terms of the underlying unitary operator $U$.
    Then, we derive an upper bound on the fraction of $U$ that result in high error probability.
    Finally, using this bound, we show that for some query number $N$ on the order of $\mathcal{O}(1/\varepsilon^2)$, the fraction becomes exponentially small.

    Assuming the given channel $\mathcal{E}_U$ satisfies $D(\mathcal{E}_U, \mathcal{E}_I)\geq\varep$, the $H_1$ error probability $P_{\text{error}|\mathcal{E}_U}$ for input Haar-random states $\boldsymbol{\psi}\coleq(\ket{\psi_1},\dots,\ket{\psi_N})$ is given as follows:
    \begin{align}
    P_{\text{error}|\mathcal{E}_U}
        &=\pr_{\boldsymbol{\psi}}(\text{Decide }\mathcal{E}_U=\mathcal{E}_I)\\
        &=\E_{\boldsymbol{\psi}}\prod_{k=1}^N \pr(X_k=0|\psi_k)\\
        &=\prod_{k=1}^N \E_{\psi_k}|\bra{\psi_k}U\ket{\psi_k}|^2\\
        &=\prod_{k=1}^N \E_{\psi_k}\tr((U\otimes U^\dagger)(\ketbra{\psi_k}{\psi_k})^{\otimes2})\\
        &=\left(\frac{d+|\tr (U)|^2}{d(d+1)}\right)^N.\label{eq:H1error}
            && \text{Eq.}~\eqref{eq:haar}
    \end{align}
    Then, the small error probability condition $P_{\text{error}|\mathcal{E}_U}<1/3$ can be equivalently written as follows:
    \begin{align}
        |\tr (U)|^2<3^{-\frac{1}{N}}d(d+1)-d.
    \end{align}
    Consequently, we can write the fraction of $U\sim\epcue$ with large error probability as follows:
    \begin{align}
        \pr_{U\sim\epcue}(P_{\text{error}|\mathcal{E}_U}\geq1/3)\equiv
        \pr_{U\sim\epcue}(|\tr (U)|^2\geq3^{-\frac{1}{N}}d(d+1)-d)
        \label{eq:tr_bound_1}
    \end{align}
    Our goal is to show that the RHS of Eq.~\eqref{eq:tr_bound_1} is on the order of $\exp(-\Omega(d))$ for a number of queries $N$, satisfying $N=\mathcal{O}(1/\varep^2)$.
    Thus, we aim to obtain its upper bound in terms of $N$.
    From
    \begin{align}
        3^{-\frac{1}{N}}d(d+1)-d
        &=d^2-(1-e^{-\frac{\log 3}{N}})(d^2+d)\\
        &> d^2-\frac{\log 3}{N}(d^2+d)
            && e^{-x}>1-x\text{ for }x>0\\
        &> d^2-\frac{2\log 3}{N}d^2\\
        &> d^2\left(1-\frac{2\log 3}{N}\right)^2,
    \end{align}
    we obtain the upper bound relation
    \begin{align}
        \pr_{U\sim\epcue}(|\tr (U)|^2\geq3^{-\frac{1}{N}}d(d+1)-d)
        &\leq\pr_{U\sim\epcue}\left(|\tr (U)|\geq d\left(1-\frac{2\log3}{N}\right)\right).
    \end{align}
    Writing the average eigenangle as $\bar{\theta}\coleq(\theta_1+\dots+\theta_d)/d$, we have
    \begin{align}
        |\tr (U)|
        &=\left|\sum_{k=1}^d e^{i\theta_k}\right|\\
        &=\sqrt{\sum_{1\leq k,l\leq d}\cos(\theta_k-\theta_l)}\\
        &\leq\sqrt{d^2-\frac{1}{3}\sum_{1\leq k,l\leq d}(\theta_k-\theta_l)^2}\\
        &=\sqrt{d^2-\frac{1}{3}\sum_{1\leq k,l\leq d}((\theta_k-\bar{\theta})-(\theta_l-\bar{\theta}))^2}\\
        &=\sqrt{d^2-\frac{1}{3}\left(2d\sum_{k=1}^d(\theta_k-\bar{\theta})^2-2\left(\sum_{k=1}^d(\theta_k-\bar{\theta})\right)^2\right)}\\
        &=d\sqrt{1-\frac{2}{3d}\sum_{k=1}^d(\theta_k-\bar{\theta})^2}\\
        &\leq d-\frac{1}{3}\sum_{k=1}^d(\theta_k-\bar{\theta})^2,
    \end{align}
    leading to another upper bound relation
    \begin{align}
        &\pr_{U\sim\epcue}\left(|\tr (U)|\geq d\left(1-\frac{2\log 3}{N}\right)\right)\notag\\
        &\quad\leq\pr_{U\sim\epcue}\left(d-\frac{1}{3}\sum_{k=1}^d(\theta_k-\bar{\theta})^2>d\left(1-\frac{2\log 3}{N}\right)\right)\\
        &\quad=\pr_{U\sim\epcue}\left(\sum_{k=1}^d(\theta_k-\bar{\theta})^2<\frac{(6\log3)d}{N}\right).
        \label{eq:tr_bound_2}
    \end{align}
    As a result, our goal reduces to deriving an upper bound on the fraction of $U$ whose eigenangles have small sample variance.

    A notable feature of CUE is that its eigenangles behave as repelling particles on a unit circle~\cite{dyson1962brownian}.
    Therefore, one can presume that the sample variance of eigenangles of $\epcue$ is larger compared to that of uniformly sampled angles.
    We rigorously prove this presumption and consequently obtain the upper bound on Eq.~\eqref{eq:tr_bound_2}.
    Before proceeding, we first define the notion of uniformly sampled angles.
    Let $\epuni$ be the distribution of $d$-dimensional angles $\btheta$, where the angle follows the pdf
    \begin{align}
        f_{\epuni}(\btheta)\coleq\frac{1}{\varep^{d-2}d(d-1)}\ind_{\btheta}(\mathcal{R}(s)).
    \end{align}
    This construction offers a fair comparison to the eigenangles from $\epcue$; 
    $\btheta\sim\epuni$ is sampled by uniformly choosing $m$ and $n$ satisfying $\theta_m=-s/2$ and $\theta_n=s/2$, and then sampling the rest of the $\theta_k$'s from the uniform distribution in $[-s/2,s/2]$.
    Now, we have the following lemma:
    \begin{lemma}
    \label{lemma:anti-concentration}
        Let $\bar{\theta}=\sum_{j=1}^d\theta_j/d$. For $\varep<1/2$, $d\geq 4$, and $\delta< ds^2/36e^{24}$, we have
        \begin{align*}
            \pr_{\btheta\sim\epcue}\left(\sum_{j=1}^d (\theta_j-\bar{\theta})^2<\delta\right)<\pr_{\btheta\sim\epuni}\left(\sum_{j=1}^d (\theta_j-\bar{\theta})^2<\delta\right).
        \end{align*}
    \end{lemma}
    \begin{proof}
    Appendix \ref{appsubsec:anti-concentration}
    \end{proof}
    \noindent
    Thus, under the assumption
    \begin{align}
        \label{eq:C_assumption}
        \frac{(6\log3)d}{N} < \frac{ds^2}{36e^{24}},
    \end{align}
    we obtain the upper bound for Eq.~\eqref{eq:tr_bound_2} as
    \begin{align}
        &\pr_{U\sim\epcue}\left(\sum_{k=1}^d(\theta_k-\bar{\theta})^2<\frac{(6\log3)d}{N}\right)\notag\\
        &\quad\equiv\pr_{\btheta\sim\epcue}\left(\sum_{k=1}^d(\theta_k-\bar{\theta})^2<\frac{(6\log3)d}{N}\right)\\
        &\quad<\pr_{\btheta\sim\epuni}\left(\sum_{k=1}^d(\theta_k-\bar{\theta})^2<\frac{(6\log3)d}{N}\right)\\
        &\quad=\pr_{\btheta\sim\epuni}\left(\sum_{k=1}^d \theta_k^2-d\bar{\theta}^2<\frac{(6\log3)d}{N}\right)\\
        &\quad\leq \pr_{\btheta\sim\epuni}\left(\sum_{k=1}^d \theta_k^2<\frac{(6\log3)d}{N}+Cd\right)+\pr_{\btheta\sim\epuni}\left(\bar{\theta}^2>C\right),
        \label{eq:tr_bound_3}
    \end{align}
    where $C$ is a constant.
    The third line follows from Lemma \ref{lemma:anti-concentration}, and the last line follows from the union bound.
    
    In the following, we obtain the upper bound on the two terms in Eq.~\eqref{eq:tr_bound_3} respectively.
    We can obtain the upper bound for the first term as
    \begin{align}
        &\pr_{\btheta\sim\epuni}\left(\sum_{k=1}^d \theta_k^2<\frac{(6\log3)d}{N}+Cd\right)\\
        &\quad=\pr_{\btheta\sim\epuni}\left(\sum_{k=1}^d \theta_k^2-\E \sum_{k=1}^{d} {\theta_k}^2<\frac{(6\log3)d}{N}+Cd-\E \sum_{k=1}^{d} {\theta_k}^2\right)\\
        &\quad=\pr_{\theta'_k\sim\text{Uniform}(-s/2,s/2)}\left(\sum_{k=1}^{d-2} ({\theta'_k}^2-\E {\theta'_k}^2) <\frac{(6\log3)d}{N}+Cd-\frac{s^2}{2}-\frac{s^2(d-2)}{12}\right)\\
        &\quad\equiv\pr_{\theta'_k\sim\text{Uniform}(-s/2,s/2)}\left(\sum_{k=1}^{d-2} ({\theta'_k}^2-\E {\theta'_k}^2) <X\right)\\
        &\quad\leq\exp\left(-\frac{32X^2}{s^4(d-2)}\right)
    \end{align}
    under the assumption $X<0$, where the variable $X$ is defined as follows:
    \begin{align}
        X\coleq \frac{(6\log3)d}{N}+Cd-\frac{s^2}{2}-\frac{s^2(d-2)}{12}.
    \end{align}
    Here, the angles $\theta_k'$'s are reindexed elements of $\btheta$, excluding the two angles $s/2$ and $-s/2$.
    The last line follows from Hoeffding's inequality.
    Similarly, we can obtain the upper bound for the second term as
    \begin{align}
        \pr_{\btheta\sim\epuni}\left(\bar{\theta}^2>C\right)
        &=\pr_{\theta'_k\sim\text{Uniform}(-s/2,s/2)}\left(\left|\sum_{k=1}^{d-2}\theta'_k\right|>\sqrt{C}d\right)\\
        &\leq2\exp\left(-\frac{2Cd^2}{s^2(d-2)}\right).
    \end{align}
    Finally, we obtain the following upper bound for the fraction of $U\sim\epcue$ with large error probability:
    \begin{align}
        \pr_{U\sim\epcue}(P_{\text{error}|\mathcal{E}_U}\geq1/3)\leq\exp\left(-\frac{32X^2}{s^4(d-2)}\right)+2\exp\left(-\frac{2Cd^2}{s^2(d-2)}\right).
    \end{align}

    We choose the parameters as follows:
    \begin{align}
        N&=\frac{217e^{24}\log3}{s^2},\\
        C&=\frac{s^2}{100}.
    \end{align}
    These parameters are valid, as they satisfy the assumptions $N=\mathcal{O}(1/\varep^2)$, Eq.~\eqref{eq:C_assumption}, and $X<0$.
    From
    \begin{align}
        X=\frac{6s^2d}{217e^{24}}+\frac{s^2d}{100}-\frac{s^2}{2}-\frac{s^2(d-2)}{12}<\frac{s^2(d-2)}{24},
    \end{align}
    we obtain
    \begin{align}
        \pr_{U\sim\epcue}(P_{\text{error}|\mathcal{E}_U}\geq1/3)
        &\leq\exp\left(-\frac{d-2}{18}\right)+2\exp\left(-\frac{d^2}{50(d-2)}\right)\\
        &=\exp(-\Omega(d)),
    \end{align}
    which completes the proof.
    Note that the numerical results in Sec.~\ref {sec:average} imply that, in practice, the number of queries $N$ will be significantly smaller than the one given in the proof.
\end{proof}

\section{Proof of technical lemmas}
    \label{appsec:lemma}
\subsection{Proof of Lemma \ref{lemma:F1F2}}
    \label{appsubsec:F1F2}
    We have
    \begin{align}
        \E_\psi\pr_{\bo\sim p_0}((G_\psi)^c)
        &\equiv\E_\psi\pr_{\bo\sim p_0}((A_\alpha\cap B_{\beta,\psi}\cap C_{\gamma,\psi})^c)\\
        &=\E_\psi\pr_{\bo\sim p_0}((A_\alpha\cap B_{\beta,\psi})^c\cup(A_\alpha\cap C_{\gamma,\psi})^c)\\
        &\leq\E_\psi\pr_{\bo\sim p_0}((A_\alpha\cap B_{\beta,\psi})^c)+\E_\psi\pr_{\bo\sim p_0}((A_\alpha\cap C_{\gamma,\psi})^c)\\
        &=1-\E_\psi\pr_{\bo\sim p_0}(A_\alpha\cap B_{\beta,\psi})+1-\E_\psi\pr_{\bo\sim p_0}(A_\alpha\cap C_{\gamma,\psi}),
    \end{align}
    where the third line follows from the union bound.
    Thus, it suffices to derive lower bounds for the terms $\E_\psi\pr_{\bo\sim p_0}(A_\alpha\cap B_{\beta,\psi})$ and $\E_\psi\pr_{\bo\sim p_0}(A_\alpha\cap C_{\gamma,\psi})$, which are the fractions of the sets $A_\alpha\cap B_{\beta,\psi}$ and $A_\alpha\cap C_{\gamma,\psi}$ in $(\psi,\bo)$-space, respectively.
    We provide the following three technical lemmas for this purpose.
    The first two lemmas give bounds for the first and second moments of $X_k(\psi, \bo)$ in $\psi$-space. 
    The last lemma shows an upper bound on the fraction of $A_\alpha$.
    Throughout the proof of the lemmas, we assume that $E_{o_{k}}^{{\bo_{<k}}}$ is a rank-1 operator and the input state $\rho^{\bo_{<k}}$ is pure, without loss of generality.
    The validation of these assumptions and details of the lemmas are given in Sec.~\ref{appsubsec:haar}.
    \begin{lemma}
    \label{lemma:X_1moment}
        For $s<1$, we have
        \begin{align*}
            \E_{\psi}X_k(\psi,\bo)\geq-\frac{s^2}{d}.
        \end{align*}
    \end{lemma}
    \begin{proof}
    Sec.~\ref{appsubsubsec:X_1moment}
    \end{proof}
    \begin{lemma}
    \label{lemma:X_2moment}
        For $s<1$, we have
        \begin{align*}
            \E_{\psi}X_k^2(\psi,\bo)\leq\frac{6s^2\left(f(E_{o_k}^{\bo_{<k}},\rho^{\bo_{<k}})+1\right)}{d^2}+\frac{72s^4\left(f^2(E_{o_k}^{\bo_{<k}},\rho^{\bo_{<k}})+1\right)}{d^4}.
        \end{align*}
    \end{lemma}
    \begin{proof}
    Sec.~\ref{appsubsubsec:X_2moment}
    \end{proof}

    \begin{lemma}
    \label{lemma:orth_POVM}
        \begin{align*}
            \pr_{\bo\sim p_0}(A_\alpha)\geq1-\frac{d}{\alpha}
        \end{align*}
    \end{lemma}
    \begin{proof}
        We show that
        \begin{align}
            \pr_{\bo\sim p_0|\bo_{<k}}\left(\sum_{k=1}^Nf(E_{o_k}^{\bo_{<k}},\rho^{\bo_{<k}})>\alpha N\right)\leq\frac{d}{\alpha}.
        \end{align}
        Recall that
        \begin{align}
            p_0(o_k|\bo_{<k})&=w_{o_{k}}^{{\bo_{<k}}}\tr \left(E_{o_{k}}^{{\bo_{<k}}}\rho^{\bo_{<k}}\right)
        \end{align}
        and
        \begin{align}
            \sum_{o_k}w_{o_{k}}^{{\bo_{<k}}}E_{o_{k}}^{{\bo_{<k}}}=I\otimes I_{\text{anc}}
        \end{align}
        holds for the POVM set $\{w_{o_{k}}^{{\bo_{<k}}}E_{o_{k}}^{{\bo_{<k}}}\}_{o_k}$.
        From Markov's inequality, we obtain
        \begin{align}
            &\pr_{\bo\sim p_0}\left(\sum_{k=1}^Nf(E_{o_k}^{\bo_{<k}},\rho^{\bo_{<k}})>\alpha N\right)\notag\\
            &\quad\leq\frac{1}{\alpha N}
            \E_{\bo\sim p_0}\sum_{k=1}^Nf(E_{o_k}^{\bo_{<k}},\rho^{\bo_{<k}})\\
            &\quad=\frac{1}{\alpha N}\sum_{\bo}p_0(\bo)\sum_{k=1}^Nf(E_{o_k}^{\bo_{<k}},\rho^{\bo_{<k}})\\
            &\quad=\frac{1}{\alpha N}\sum_{k=1}^N\sum_{o_k}p_0(o_k)f(E_{o_k}^{\bo_{<k}},\rho^{\bo_{<k}})\\
            &\quad=\frac{1}{\alpha N}\sum_{k=1}^N\sum_{\bo_{<k}}\sum_{o_k|\bo_{<k}}p_0(o_k|\bo_{<k})p_0(\bo_{<k})f(E_{o_k}^{\bo_{<k}},\rho^{\bo_{<k}})\\
            &\quad=\frac{1}{\alpha N}\sum_{k=1}^N\sum_{\bo_{<k}}p_0(\bo_{<k})\sum_{o_k|\bo_{<k}}p_0(o_k|\bo_{<k})f(E_{o_k}^{\bo_{<k}},\rho^{\bo_{<k}})\\
            &\quad=\frac{1}{\alpha N}\sum_{k=1}^N\sum_{\bo_{<k}}p_0(\bo_{<k})\sum_{o_k|\bo_{<k}}w_{o_{k}}^{{\bo_{<k}}}\tr \left(E_{o_{k}}^{{\bo_{<k}}}\rho^{\bo_{<k}}\right)\frac{\tr\left(\tr_S(E_{o_{k}}^{{\bo_{<k}}})\tr_S(\rho^{\bo_{<k}})\right)}{\tr\left(E_{o_{k}}^{{\bo_{<k}}}\rho^{\bo_{<k}}\right)}\\
            &\quad=\frac{1}{\alpha N}\sum_{k=1}^N\sum_{\bo_{<k}}p_0(\bo_{<k})\tr\left(\tr_S\left(\sum_{o_k|\bo_{<k}}w_{o_{k}}^{{\bo_{<k}}}E_{o_{k}}^{{\bo_{<k}}}\right)\tr_S(\rho^{\bo_{<k}})\right)\\
            &\quad=\frac{1}{\alpha N}\sum_{k=1}^N\sum_{\bo_{<k}}p_0(\bo_{<k})\tr\left(\tr_S\left(I\otimes I_{\text{anc}}\right)\tr_S(\rho^{\bo_{<k}})\right)\\
            &\quad=\frac{d}{\alpha},
        \end{align}
        which completes the proof.
    \end{proof}
    
    \noindent We now derive the lower bound of the fraction of the sets $A_\alpha\cap B_{\beta,\psi}$ and $A_\alpha\cap C_{\gamma,\psi}$.
    We start with a lower bound of the fraction of $A_\alpha$, given as follows:
    \begin{align}
        \pr_{\bo\sim p_0}(A_\alpha)
        &\geq 1-\frac{d}{\alpha} 
            && \text{Lemma \ref{lemma:orth_POVM}}\\
        &=0.99.
            && \alpha=100d \label{eq:A_alpha}
    \end{align}
    Next, we obtain the lower bound of the fraction of $A_{\alpha}\cap B_{\beta,\psi}$ by combining the following three inequalities.
    First, denoting 
    \begin{align}
        \mu_k(\psi)&\coleq\E_{\bo\sim p_0|A_\alpha}X_k(\psi,\bo),\\
        \sigma_k^2(\psi)&\coleq\text{Var}_{\bo\sim p_0|A_\alpha}X_k(\psi,\bo),
    \end{align}
    we obtain
    \begin{align}
        &\pr_{\bo\sim p_0|A_\alpha}(X_k(\psi,\bo)<-\beta)\notag\\
        &\quad=\pr_{\bo\sim p_0|A_\alpha}(X_k(\psi,\bo)-\mu_k(\psi)<-\beta-\mu_k(\psi))\\
        &\quad=\pr_{\bo\sim p_0|A_\alpha}(X_k(\psi,\bo)-\mu_k(\psi)<-\beta-\mu_k(\psi))(\ind_\psi(\mu_k(\psi)\geq-\beta/2)+\ind_\psi(\mu_k(\psi)<-\beta/2))\\
        &\quad\leq\pr_{\bo\sim p_0|A_\alpha}(|X_k(\psi,\bo)-\mu_k(\psi)|>\beta/2)\ind_\psi(\mu_k(\psi)\geq-\beta/2)+\ind_\psi(\mu_k(\psi)<-\beta/2) \\
        &\quad\leq\pr_{\bo\sim p_0|A_\alpha}(|X_k(\psi,\bo)-\mu_k(\psi)|>\beta/2)+\ind_\psi(\mu_k(\psi)<-\beta/2) \\
        &\quad\leq\frac{4\sigma_k^2(\psi)}{\beta^2}+\ind_\psi(\mu_k(\psi)<-\beta/2), \label{eq:X|A}
    \end{align}
    where the last line holds from Chebyshev's inequality.
    Here, we define the conditional probability $\pr_{\bo\sim p|A}(B)$ for a distribution $p$, a set of the measurement outcome $A$, and an event $B$ as follows:
    \begin{align}
        \pr_{\bo\sim p|A}(B)=\frac{\sum_{\bo\in A\cap B}p(\bo)}{\sum_{\bo\in A}p(\bo)}.
    \end{align}
    Second, for $\bo\in A_\alpha$, we obtain
    \begin{align}
        &\sum_{k=1}^N\E_\psi X_k^2(\psi,\bo)\notag\\
        &\quad\leq\frac{6s^2}{d^2}\sum_{k=1}^N\left(f(E_{o_{k}}^{\bo_{<k}},\rho^{\bo_{<k}})+1\right)+\frac{72s^4}{d^4}\sum_{k=1}^N\left(f^2(E_{o_{k}}^{\bo_{<k}},\rho^{\bo_{<k}})+1\right)
            && \text{Lemma }\ref{lemma:X_2moment}\\
        &\quad\leq\frac{6s^2}{d^2}\sum_{k=1}^Nf(E_{o_{k}}^{\bo_{<k}},\rho^{\bo_{<k}}) +\frac{6s^2N}{d^2}+\frac{72 s^4}{ d^4}\left(\sum_{k=1}^Nf(E_{o_{k}}^{\bo_{<k}},\rho^{\bo_{<k}})\right)^2+\frac{72 s^4N}{ d^4}\\
        &\quad\leq\frac{606s^2N}{d}+\frac{720072 s^4N^2}{d^2}\\
        &\quad\equiv g(s, d, N),\label{eq:exp_X^2}
    \end{align}
    where the third line follows from $\sum_{k=1}^Nf(E_{o_{k}}^{\bo_{<k}},\rho^{\bo_{<k}})\leq100dN$.
    Lastly, to handle the case $\mu_k(\psi)<-\beta/2$ in Eq.~\eqref{eq:X|A}, we obtain
    \begin{align}
        &\sum_{k=1}^N\pr_\psi(\mu_k(\psi)<-\beta/2)\notag\\
        &\quad=\sum_{k=1}^N\pr_\psi(\mu_k(\psi)-\E_\psi\mu_k(\psi)<-\beta/2-\E_\psi\mu_k(\psi))\\
        &\quad=\sum_{k=1}^N\pr_\psi(\mu_k(\psi)-\E_\psi\mu_k(\psi)<-\beta/2-\E_{\bo\sim p_0|A_\alpha}\E_\psi X_k(\psi,\bo))\\
        &\quad\leq\sum_{k=1}^N\pr_\psi\left(\mu_k(\psi)-\E_\psi\mu_k(\psi)<-\beta/2+\frac{ s^2}{d}\right)
            && \text{Lemma \ref{lemma:X_1moment}}\\
        &\quad\leq\sum_{k=1}^N\pr_\psi\left(\mu_k(\psi)-\E_\psi\mu_k(\psi)<-\beta/4\right)
            && \text{Eq.~\eqref{eq:beta_condition}}\\
        &\quad\leq\sum_{k=1}^N\pr_\psi\left(|\mu_k(\psi)-\E_\psi\mu_k(\psi)|>\beta/4\right)\\
        &\quad\leq\sum_{k=1}^N\frac{16(\E_\psi\mu_k^2(\psi)-(\E_\psi\mu_k(\psi))^2)}{\beta^2}
            && \text{Chebyshev's inequality}\\
        &\quad\leq\sum_{k=1}^N\frac{16\E_\psi\mu_k^2(\psi)}{\beta^2}\\
        &\quad=\sum_{k=1}^N\frac{16\E_\psi(\E_{\bo\sim p_0|A_\alpha}X_k(\psi,\bo))^2}{\beta^2}\\
        &\quad\leq\sum_{k=1}^N\frac{16\E_{\bo\sim p_0|A_\alpha}\E_\psi X_k^2(\psi,\bo)}{\beta^2}\\
        &\quad=\frac{16\E_{\bo\sim p_0|A_\alpha}\sum_{k=1}^N\E_\psi X_k^2(\psi,\bo)}{\beta^2}\\
        &\quad\leq\frac{16g(s,d,N)}{\beta^2}. \label{eq:mu}
            && \text{Eq.~\eqref{eq:exp_X^2}}
    \end{align}
    Summing up, we obtain the lower bound on the fraction of $A_\alpha\cap B_{\beta,\psi}$ as follows:
    \begin{align}
        &\E_\psi\pr_{\bo\sim p_0}(A_\alpha\cap B_{\beta,\psi})\notag\\
        &\quad=\pr_{\bo\sim p_0}(A_\alpha)\E_\psi\pr_{\bo\sim p_0|A_\alpha}(B_{\beta,\psi})\\
        &\quad\geq0.99\E_\psi\pr_{\bo\sim p_0|A_\alpha}(B_{\beta,\psi})
            && \text{Eq.~\eqref{eq:A_alpha}}\\
        &\quad=0.99\E_\psi\pr_{\bo\sim p_0|A_\alpha}(X_k(\psi,\bo)\geq-\beta\text{ for all }1\leq k\leq N)\\
        &\quad=0.99(1-\E_\psi\pr_{\bo\sim p_0|A_\alpha}(\exists k\textit{ s.t. }X_k(\psi,\bo)<-\beta))\\
        &\quad\geq0.99\left(1-\E_\psi\sum_{k=1}^N\pr_{\bo\sim p_0|A_\alpha}(X_k(\psi,\bo)<-\beta)\right)
            && \text{Union bound}\\
        &\quad\geq0.99\left(1-\E_\psi\sum_{k=1}^N\left(\frac{4\sigma_k^2(\psi)}{\beta^2}+\ind_\psi(\mu_k(\psi)<-\beta/2)\right)\right)
            && \text{Eq.~\eqref{eq:X|A}}\\
        &\quad=0.99\left(1-\sum_{k=1}^N\left(\frac{4\E_\psi\text{Var}_{\bo\sim p_0|A_\alpha}X_k(\psi,\bo)}{\beta^2}+\pr_\psi(\mu_k(\psi)<-\beta/2)\right)\right)\\
        &\quad\geq0.99\left(1-\sum_{k=1}^N\frac{4\E_{\bo\sim p_0|A_\alpha}\E_\psi X_k^2(\psi, \bo)}{\beta^2}-\frac{16g(s,d,N)}{\beta^2}\right)
            && \text{Eq.~\eqref{eq:mu}}\\
        &\quad\geq0.99\left(1-\frac{20g(s,d,N)}{\beta^2}\right)
            && \text{Eq.~\eqref{eq:exp_X^2}}\\
        &\quad\geq 1-\left(0.01+\frac{20g(s,d,N)}{\beta^2}\right)\\
        &\quad\equiv1-F_1 \label{eq:F1}
    \end{align}
    Finally, we obtain the lower bound on the fraction of $A_\alpha\cap C_{\gamma,\psi}$ as follows:
    \begin{align}
        &\E_\psi\pr_{\bo\sim p_0}(A_\alpha\cap C_{\gamma,\psi})\notag\\
        &\quad=\pr_{\bo\sim p_0}(A_\alpha)\E_\psi\pr_{\bo\sim p_0|A_\alpha}(C_{\gamma,\psi})\\
        &\quad\geq0.99\E_\psi\pr_{\bo\sim p_0|A_\alpha}(C_{\gamma,\psi})
            && \text{Eq.~\eqref{eq:A_alpha}}\\
        &\quad=0.99\E_\psi\pr_{\bo\sim p_0|A_\alpha}\left(\sum_{k=1}^N \E_{\bo\sim p_0|\bo_{<k}}X_k^2(\psi,\bo_{\leq k})\leq \gamma\right)\\
        &\quad=0.99\left(1-\E_\psi\pr_{\bo\sim p_0|A_\alpha}\left(\sum_{k=1}^N \E_{\bo\sim p_0|\bo_{<k}}X_k^2(\psi,\bo_{\leq k})> \gamma\right)\right)\\
        &\quad\geq0.99\left(1-\frac{\sum_{k=1}^N \E_{\bo\sim p_0|A_\alpha}\E_{\bo\sim p_0|\bo_{<k}}\E_\psi X_k^2(\psi,\bo_{\leq k})}{\gamma}\right)
            && \text{Markov's inequality}\\
        &\quad\geq0.99\left(1-\frac{g(s,d,N)}{\gamma}\right)
            && \text{Eq.~\eqref{eq:exp_X^2}}\\
        &\quad\geq1-\left(0.01+\frac{g(s,d,N)}{\gamma}\right)\\
        &\quad\equiv1-F_2. \label{eq:F2}
    \end{align}
    We collect the inequalities and obtain the following upper bound:
    \begin{align}
        \E_\psi\pr_{\bo\sim p_0}((G_\psi)^c)\leq F_1+F_2.\label{eq:F1F2}
    \end{align}

\subsection{Proof of Lemma \ref{lemma:F3F4}}
    \label{appsubsec:F3F4}
    We employ the following Martingale concentration lemma:
    \begin{lemma} (\cite{chen2023efficient})
        \label{lemma:martingale}
        Let 
        \begin{align*}
            Y_k(\psi,\bo)=\log(1+X_k(\psi,\bo))\ind_{\bo}( B_{\beta,\psi}).
        \end{align*}
        Then for any $\gamma,\eta>0$, we have
        \begin{align}
            \pr_{\bo}\left(\sum_{k=1}^N Y_k(\psi,\bo)\leq-\left(1+\frac{1}{\beta}\right)\sum_{k=1}^N \E_{\bo\sim p_0|\bo_{<k}}X_k^2(\bo_{\leq k})-\eta\text{ and }\sum_{k=1}^N \E_{\bo\sim p_0|\bo_{<k}}X_k^2(\bo_{\leq k})\leq\gamma\right)
            &\leq\exp\left(-\frac{\eta^2}{4\gamma+2\beta\eta/3}\right).
        \end{align}
    \end{lemma}    
    \noindent From
    \begin{align}
        &\pr_{\bo\sim p_0|G_\psi}\left(\sum_{k=1}^N Y_k(\psi,\bo)\leq-\left(1+\frac{1}{\beta}\right)\sum_{k=1}^N \E_{\bo\sim p_0|\bo_{<k}}X_k^2(\bo_{\leq k})-\eta\text{ and }\sum_{k=1}^N \E_{\bo\sim p_0|\bo_{<k}}X_k^2(\bo_{\leq k})\leq\gamma\right)\notag\\
        &\quad\geq\pr_{\bo\sim p_0|G_\psi}\left(\sum_{k=1}^N Y_k(\psi,\bo)\leq-\left(1+\frac{1}{\beta}\right)\gamma-\eta\right)
            && \bo\in C_{\gamma,\psi}\\
        &\quad=\pr_{\bo\sim p_0|G_\psi}\left(\sum_{k=1}^N \log(1+X_k(\psi,\bo))\leq-\left(1+\frac{1}{\beta}\right)\gamma-\eta\right)
            && \bo\in B_{\beta,\psi}\\
        &\quad=\pr_{\bo\sim p_0|G_\psi}\left(L(\psi,\bo)\leq\exp\left(-\left(1+\frac{1}{\beta}\right)\gamma-\eta\right)\right)
    \end{align}
    and Lemma \ref{lemma:martingale}, we obtain the probabilistic upper bound of $L(\psi,\bo)$ as
    \begin{align}
        \pr_{\bo\sim p_0|G_\psi}\left(L(\psi,\bo)\leq\exp\left(-\left(1+\frac{1}{\beta}\right)\gamma-\eta\right)\right)
        &\leq\exp\left(-\frac{\eta^2}{4\gamma+2\beta\eta/3}\right)
    \end{align}
    or equivalently,
    \begin{align}
    \label{eq:martingale_result}
        \pr_{\bo\sim p_0|G_\psi}\left(L(\psi,\bo)\leq1-F_3\right)
        &\leq F_4
    \end{align}
    with
    \begin{align}
        F_3&\equiv1-\exp\left(-\left(1+\frac{1}{\beta}\right)\gamma-\eta\right),\\
        F_4&\equiv\exp\left(-\frac{\eta^2}{4\gamma+2\beta\eta/3}\right).
    \end{align}
    From this, we obtain the upper bound as follows:
    \begin{align}
        &\E_\psi\E_{\bo\sim p_0|G_\psi}\max(0, 1- L(\psi,\bo))\notag\\
        &\quad=\E_\psi\E_{\bo\sim p_0|G_\psi}\max(0, 1- L(\psi,\bo))\ind_{\psi,\bo}(L(\psi,\bo)>1-F_3)\notag\\
        &\quad\quad+\E_\psi\E_{\bo\sim p_0|G_\psi}\max(0, 1- L(\psi,\bo))\ind_{\psi,\bo}(L(\psi,\bo)\leq1-F_3)\\
        &\quad< \E_\psi\E_{\bo\sim p_0|G_\psi}F_3\ind_{\psi,\bo}(L(\psi,\bo)>1-F_3)+\E_\psi\E_{\bo\sim p_0|G_\psi}\ind_{\psi,\bo}(L(\psi,\bo)\leq1-F_3)
            &&L(\psi,\bo)\geq0\\
        &\quad= \E_\psi F_3\pr_{\bo\sim p_0|G_\psi}(L(\psi,\bo)>1-F_3)+\E_\psi\pr_{\bo\sim p_0|G_\psi}(L(\psi,\bo)\leq1-F_3)\\
        &\quad\leq F_3+F_4.\label{eq:F3F4}
            &&\text{Eq.~\eqref{eq:martingale_result}}
    \end{align}

\subsection{Proof of Lemma \ref{lemma:qsvtpoly}}
\label{appsubsec:qsvtpoly}

Conditions (1), (2), and (4) are straightforward.
Thus, we show that conditions (3) and (5) are satisfied.
We first show that condition (5) is satisfied.
The Chebyshev polynomial has the following definition:
\begin{align}
    T_n(x)&\coleq\frac{1}{2}\left(\left(x-\sqrt{x^2-1}\right)^n+\left(x+\sqrt{x^2-1}\right)^n\right).
\end{align}
For even $n$, we can rewrite the definition as follows:
\begin{align}
    T_n(x)&\equiv
    \begin{cases}
        \cos(n\cos^{-1} x) &|x|\leq 1\\
        \cosh(n\cosh^{-1} |x|) &|x|>1
    \end{cases}
\end{align}
The first form implies that for $x\in[0,1-\delta]$, the numerator of $P(x)$ is bounded by 1 as $|T_n(x/(1-\delta))|\leq 1$.
Thus, showing that the denominator is sufficiently large, \textit{i.e.} $|T_n(1/(1-\delta)|\geq1/\Delta$, is sufficient for condition (5). 
From $\cosh x\equiv(e^x+e^{-x})/2$, we have
\begin{align}
    T_n(1/(1-\delta))
    &\equiv\cosh(n\cosh^{-1}(1/(1-\delta)))\\
    &>\frac{1}{2}\exp(n\cosh^{-1}(1/(1-\delta))).
\end{align}
The Taylor expansion gives us
\begin{align}
    \cosh \sqrt{x}&=1+\frac{x}{2!}+\frac{x^2}{4!}+\dots\\
    &\leq1+x
\end{align}
for $x\in[0,1]$, leading to
\begin{align}
    \cosh^{-1}(1/(1-\delta))
    &\geq\cosh^{-1}(1+\delta)\\
    &\geq\sqrt{\delta},
\end{align}
where the first line follows from the monotonicity of $\cosh^{-1}.$
Thus, we have the lower bound of the denominator $|T_n(1/(1-\delta)|$ as
\begin{align}
    \frac{1}{2}\exp(n\cosh^{-1}(1/(1-\delta))
    &\geq\frac{1}{2}\exp\left(n\sqrt{\delta}\right)\\
    &=\frac{1}{2}\exp\left(2\ceil{1/\sqrt{\delta}\log(2/\Delta)}\sqrt{\delta}\right)\\
    &\geq\frac{1}{2}\exp\left(\log(2/\Delta)\right)\\
    &=\frac{1}{\Delta},
\end{align}
which establishes the condition (5).

We now show that condition (3) is satisfied.
In the domain $x\in[0,1-\delta]$, condition (5) implies condition (3).
Thus, it is sufficient to show that condition (3) holds on $x\in(1-\delta, 1]$.
In such a regime, the Chebyshev polynomial $T_n(x)$ is an increasing function of $x$, as $\cosh(x)$ and $\cosh^{-1}(x)$ are both monotonically increasing functions.
Therefore, $P(x)$ is an increasing function, leading to $0\leq P(x)\leq P(1)=1$.
Thus, $|P(x)|\leq1$ for $x\in[0,1]$ holds.
Since $P$ is even, this also holds in $x\in[-1,1]$, thereby establishing condition (3).

\subsection{Proof of Lemma \ref{lemma:qsvtdelta}}
\label{appsubsec:qsvtdelta}

We show that
\begin{align}
    \pr_\psi(1-|\bra{\psi}U^\dagger\ket{\psi}|<\delta)\leq\frac{8d\delta}{\varep^2}
\end{align}
by finding a lower bound of the random variable $1-|\bra{\psi}U^\dagger\ket{\psi}|$.
Since $\ket{\psi}$ is a Haar-random state, we can consider $U^\dagger$ as a diagonal operator with $e^{i\theta_1},\dots e^{i\theta_d}$ as the diagonal elements without loss of generality.
Thus, we can write $U^\dagger$ in the computational basis as follows:
\begin{align}
    U^\dagger=\sum_{j=1}^d e^{i\theta_j}\ketbra{j}{j}.
\end{align}
Then, writing $c_j\coleq|\braket{\psi}{j}|^2$, we have
\begin{align}
    |\bra{\psi}U^\dagger\ket{\psi}|
    &=\left|\sum_{j=1}^d c_je^{i\theta_j}\right|.
\end{align}
From Lemma \ref{lemma:distance}, we know that there exist $\theta_k$ and $\theta_l$ for some $1\leq k,l\leq d$, which has a gap no smaller than $s=2\sin^{-1}(\varep/2)$ on the unit circle.
Collecting these, we have
\begin{align}
    1-|\bra{\psi}U^\dagger\ket{\psi}|
    &=1-\left|\sum_{j\neq k,l} c_je^{i\theta_j}+re^{i\phi}\right|\\
    &\geq1-\left|\sum_{j\neq k,l} c_je^{i\phi}+re^{i\phi}\right|\\
    &=1-|(1-c_k-c_l)e^{i\phi}+re^{i\phi}|\\
    &=c_k+c_l-r\\
    &=c_k+c_l-\sqrt{c_k^2+c_l^2+2c_kc_l\cos(\theta_k-\theta_l)}\\
    &=\frac{(c_k+c_l)^2-({c_k^2+c_l^2+2c_kc_l\cos(\theta_k-\theta_l)})}{c_k+c_l+\sqrt{c_k^2+c_l^2+2c_kc_l\cos(\theta_k-\theta_l)}}\\
    &=\frac{{2c_kc_l(1-\cos(\theta_k-\theta_l))}}{c_k+c_l+\sqrt{c_k^2+c_l^2+2c_kc_l\cos(\theta_k-\theta_l)}}\\
    &\geq\frac{c_kc_l}{c_k+c_l}(1-\cos(\theta_k-\theta_l))\\
    &\geq\frac{\min(c_k,c_l)}{2}\left(2\sin^2\left(\frac{\theta_k-\theta_l}{2}\right)\right)\\
    &\geq\frac{\min(c_k,c_l)}{4}\varep^2.
\end{align}
Thus, we obtain
\begin{align}
    \pr_\psi(1-|\bra{\psi}U^\dagger\ket{\psi}|<\delta)
    &\leq\pr_\psi\left(\min(c_k,c_l)\leq\frac{4\delta}{\varep^2}\right)\\
    &\leq\pr_\psi\left(c_k\leq\frac{4\delta}{\varep^2}\right)+\pr_\psi\left(c_l\leq\frac{4\delta}{\varep^2}\right)\\
    &=2\left(1-\left(1-\frac{4\delta}{\varep^2}\right)^{d-1}\right)\\
    &\leq\frac{8d\delta}{\varep^2},
\end{align}
where the second line follows from the union bound and the third line follows from $c_k,c_l\sim \text{Beta}(1,d-1)$~\cite{zyczkowski2001induced}.

\subsection{Proof of Lemma \ref{lemma:anti-concentration}}
\label{appsubsec:anti-concentration}
    We show that for $\varep<1/2$, $d\geq 4$, and $\delta< ds^2/36e^{24}$,
    \begin{align}
        \pr_{\btheta\sim\epcue}\left(\sum_{j=1}^d(\theta_j-\bar{\theta})^2<\delta\right)<\pr_{\btheta\sim\epuni}\left(\sum_{j=1}^d (\theta_j-\bar{\theta})^2<\delta\right)
    \end{align}
    holds.
    Define the regime of eigenangles with small sample variance as
    \begin{align}
        \mathcal{R}_{<\delta}(s)\coleq\left\{\btheta:\sum_{j=1}^d (\theta_j-\bar{\theta})^2<\delta,\right\}\cap\mathcal{R}(s),
    \end{align}
    where $\mathcal{R}(s)\coleq \{\btheta:\min\btheta=-s/2,\max\btheta=s/2\}$.
    Then, we can write the error probabilities with the following integration forms:
    \begin{align}
        \pr_{\btheta\sim\epcue}\left(\sum_{j=1}^d(\theta_j-\bar{\theta})^2<\delta\right)&=\int_{\mathcal{R}_{<\delta}(s)}d\btheta f_{\epcue}(\btheta),\\
        \pr_{\btheta\sim\epuni}\left(\sum_{j=1}^d(\theta_j-\bar{\theta})^2<\delta\right)&=\int_{\mathcal{R}_{<\delta}(s)}d\btheta f_{\epuni}(\btheta).
    \end{align}
    Thus, showing $f_{\epcue}(\btheta)<f_{\epuni}(\btheta)$ for all $\btheta\in\mathcal{R}_{<\delta}(s)$ is sufficient for the proof. 
    Recall that we have the pdf's of
    \begin{align}
        f_{\epcue}(\btheta)&\coleq\frac{1}{C_\varep'}\prod_{1\le k<l\le d}|e^{i\theta_k}-e^{i\theta_l}|^2\ind_{\btheta}({\mathcal{R}(s)}),\\
        f_{\epuni}(\btheta)&\coleq\frac{1}{\varep^{d-2}d(d-1)}\ind_{\btheta}(\mathcal{R}(s)),
    \end{align}
    where $f_{\epuni}(\btheta)$ is constant for $\btheta\in\mathcal{R}_{<\delta}$.
    Thus, we aim to obtain a smaller constant upper bound of $f_{\epcue}(\btheta)$. 
    To this end, we obtain a lower bound of the denominator term $C_\varep'$ and an upper bound of the rest of the numerator term in $f_{\epcue}(\btheta)$, respectively.

    We derive the lower bound of the denominator term $C_\varep'$.
    Let 
    \begin{align}
        \mathcal{R}^{m,n}(s,\kappa)&\coleq\{\btheta:\theta_m=-s/2, \theta_n=s/2, \theta_{j\neq m,n}\in[-(s-\kappa)/2, (s-\kappa)/2]\}.
    \end{align}
    From mutually disjoint property of $\mathcal{R}^{m,n}(s,s/2)$'s, we obtain a lower bound of $C_\varep'$ as follows:
    \begin{align}
        C_\varep'&=\int_{\mathcal{R}(s)}d\btheta\prod_{1\le k<l\le d}|e^{i\theta_k}-e^{i\theta_l}|^2\\
        &>\sum_{mn}\int_{\mathcal{R}^{m,n}(s,s/2)}d\btheta\prod_{1\le k<l\le d}|e^{i\theta_k}-e^{i\theta_l}|^2\\
        &=d(d-1)\int_{[-s/4, s/4]^{d-2}}d\btheta'|e^{-is/2}-e^{is/2}|^2\prod_{1\le k\le d-2}|e^{-is/2}-e^{i\theta'_k}|^2|e^{is/2}-e^{i\theta'_k}|^2\prod_{1\le k<l\le d-2}|e^{i\theta'_k}-e^{i\theta'_l}|^2\\
        &\geq d(d-1)|e^{-is/2}-e^{is/2}|^2|e^{-is/2}-e^{-is/4}|^{2d-4}|e^{is/2}-e^{is/4}|^{2d-4}\int_{[-s/4, s/4]^{d-2}}d\btheta'\prod_{1\le k<l\le d-2}|e^{i\theta'_k}-e^{i\theta'_l}|^2\\
        &\geq d(d-1)\left(1-\frac{s^2}{24}\right)^{d(d-1)}s^2\left(\frac{s}{4}\right)^{4d-8}\int_{[-s/4, s/4]^{d-2}}d\btheta'\prod_{1\le k<l\le d-2}|\theta'_k-\theta'_l|^2\\
        &=d(d-1)\left(1-\frac{s^2}{24}\right)^{d(d-1)}s^2\left(\frac{s}{4}\right)^{4d-8}\left(\frac{s}{2}\right)^{(d-2)^2}\prod_{j=0}^{d-3}\frac{j!^2(j+1)!}{(j+d-2)!}\\
        &=d(d-1)s^{d^2-2}\left(1-\frac{s^2}{24}\right)^{d(d-1)}\left(\frac{1}{2}\right)^{d^2-4d-12}\prod_{j=0}^{d-3}\frac{j!^2(j+1)!}{(j+d-2)!}\\
        &>d(d-1) s^{d^2-2}\left(\frac{1}{3}\right)^{d(d-1)}\prod_{j=0}^{d-3}\frac{j!^2(j+1)!}{(j+d-2)!}.\label{eq:anti_1}
    \end{align}
    Here, $\btheta'\coleq(\theta_1',\dots,\theta'_{d-2})$ is a vector of reindexed eigenangles, which excludes the maximum and minimum eigenangles, $s/2$ and $-s/2$.
    The fifth line follows from two inequalities: $s<1$ (which holds from $\varep<1/2$), and the bound of
    \begin{align}
        |e^{i\theta_k}-e^{i\theta_l}|=\left|2\sin\frac{\theta_k-\theta_l}{2}\right|\geq2\left(1-\frac{1}{6}\left|\frac{\theta_k-\theta_l}{2}\right|^2\right)\left|\frac{\theta_k-\theta_l}{2}\right|\geq\left(1-\frac{s^2}{24}\right)|\theta_k-\theta_l|,
    \end{align}
    which holds under $\theta_k, \theta_l\in[-s/2,s/2]$ and $s<1$. 
    The sixth line is a result of the Selberg integral~\cite{askey1980selberg}.

    We consequently derive the upper bound of the numerator term in $f_{\epcue}(\btheta)$.
    We obtain
    \begin{align}
        \prod_{1\le k<l\le d}|e^{i\theta_k}-e^{i\theta_l}|^2
        &\leq\prod_{1\le k<l\le d}((\theta_k-\bar{\theta})-(\theta_l-\bar{\theta}))^2
            && \theta_k, \theta_l\in[-1/2,1/2]\\
        &\leq\prod_{1\le k<l\le d}2((\theta_k-\bar{\theta})^2+(\theta_l-\bar{\theta})^2)\\
        &\leq\left(\frac{2(d-1)\sum_{k=1}^d (\theta_k-\bar{\theta})^2}{d(d-1)/2}\right)^{\frac{d(d-1)}{2}}
            &&\text{AM-GM inequality}\\
        &<2^{d(d-1)}\left(\frac{1}{d}\right)^{\frac{d(d-1)}{2}}\delta^{\frac{d(d-1)}{2}}.
            &&\sum_{k=1}^d (\theta_j-\bar{\theta})^2<\delta\text{ for }\btheta\in\mathcal{R}_{<\delta}(s)\label{eq:anti_2}
    \end{align}

    Combining the two bounds on the denominator and numerator, we obtain the upper bound of $f_{\epcue}(\btheta)$ as follows:
    \begin{align}
        f_{\epcue}(\btheta)&=\frac{1}{C_\varep'}\prod_{1\le k<l\le d}|e^{i\theta_k}-e^{i\theta_l}|^2\\
        &<\left(d(d-1) s^{d^2-2}\left(\frac{1}{3}\right)^{d(d-1)}\prod_{j=0}^{d-3}\frac{j!^2(j+1)!}{(j+d-2)!}\right)^{-1}2^{d(d-1)}\left(\frac{1}{d}\right)^{\frac{d(d-1)}{2}}\delta^{\frac{d(d-1)}{2}}\\
        &=\frac{1}{s^{d-2}d(d-1)}\left(\frac{1}{d}\right)^{\frac{d(d-1)}{2}}\left(\frac{1}{s}\right)^{d(d-1)}\delta^{\frac{d(d-1)}{2}}6^{d(d-1)}\prod_{j=0}^{d-3}\frac{(j+d-2)!}{j!^2(j+1)!}.\label{eq:anti_3}
    \end{align}
    We complete our proof by showing that the RHS of Eq.~\eqref{eq:anti_3} is bounded by $f_{\epuni}(\btheta)$ under $d\geq 4$ and $\delta< ds^2/36e^{24}$.
    Employing
    \begin{align}
        n\log n-n+1<&\sum_{k=1}^n \log k<(n+1)\log (n+1)-n,\\
        \frac{1}{2}n^2\log n-\frac{1}{4}(n^2-1)<&\sum_{k=1}^n k\log k<\frac{1}{2}(n+1)^2\log (n+1)-\frac{1}{4}(n^2+2n),
    \end{align}
    we obtain the upper bound of the product term in the RHS of Eq.~\eqref{eq:anti_3} by
    \begin{align}
        &\log \prod_{j=0}^{d-3}\frac{(j+d-2)!}{j!^2(j+1)!}\notag\\
        &\quad=\sum_{j=0}^{d-3}\left(\log (j+d-2)!-2\log j!-\log (j+1)!\right)\\
        &\quad=\sum_{j=0}^{d-3}\left(\sum_{k=1}^{j+d-2}\log k-2\sum_{k=1}^{j}\log k-\sum_{k=1}^{j+1}\log k\right)\\
        &\quad=\sum_{k=1}^{d-2}\left(d-2-2(d-k-2)-(d-k-1)\right)\log k+\sum_{k=d-1}^{2d-5}(2d-k-4)\log k\\
        &\quad=\sum_{k=1}^{d-2}\left(3k-2d+3\right)\log k+\sum_{k=d-1}^{2d-5}(2d-k-4)\log k\\
        &\quad=-(4d-7)\sum_{k=1}^{d-2} \log k +(2d-4)\sum_{k=1}^{2d-5} \log k+4\sum_{k=1}^{d-2} k\log k-\sum_{k=1}^{2d-5} k\log k\\
        &\quad<-(4d-7)((d-2)\log(d-2)-d+3)+(2d-4)((2d-4)\log(2d-4)-2d+5)\notag\\
        &\quad\quad+4\left(\frac{1}{2}(d-1)^2\log (d-1)-\frac{1}{4}(d^2-2d)\right)-\left(\frac{1}{2}(2d-5)^2\log (2d-5)-\frac{1}{4}((2d-5)^2-1)\right)\\
        &\quad=-(4d-7)(d-2)\log(d-2)+(2d-4)^2\log(2d-4)\notag\\
        &\quad\quad+2(d-1)^2\log(d-1)-\frac{1}{2}(2d-5)^2\log(2d-5)-4d+7.
    \end{align}
    From
    \begin{align}
        -(4d-7)(d-2)\log(d-2)+(2d-4)^2\log(2d-4)
        &<(4d-7)(d-2)(\log(2d-4)-\log(d-2))\\
        &<4d(d-1)
    \end{align}
    and
    \begin{align}
        2(d-1)^2\log(d-1)-\frac{1}{2}(2d-5)^2\log(2d-5)
        &=\left(6d-\frac{21}{2}\right)\log(d-1)+\frac{1}{2}(2d-5)^2\log\left(\frac{d-1}{2d-5}\right)\\
        &<6d(d-1)+2d(d-1)\\
        &=8d(d-1),
    \end{align}
    we obtain
    \begin{align}
        \log \prod_{j=0}^{d-3}\frac{(j+d-2)!}{j!^2(j+1)!}<12d(d-1)
    \end{align}
    for $d\geq 4$. 
    Hence, Eq.~\eqref{eq:anti_3} leads to
    \begin{align}
        f_{\epcue}(\btheta)
        &<\frac{1}{s^{d-2}d(d-1)}\left(\frac{1}{d}\right)^{\frac{d(d-1)}{2}}\left(\frac{1}{s}\right)^{d(d-1)}\delta^{\frac{d(d-1)}{2}}6^{d(d-1)}\prod_{j=0}^{d-3}\frac{(j+d-2)!}{j!^2(j+1)!}\\
        &<\frac{1}{\varep^{d-2}d(d-1)}\left(\frac{36e^{24}\delta}{ds^2}\right)^{\frac{d(d-1)}{2}}\\
        &<\frac{1}{\varep^{d-2}d(d-1)}\\
        &=f_{\epuni}(\btheta)
    \end{align}
    for $\delta< ds^2/(36e^{24})$.

\subsection{Proof of technical lemmas on Haar randomness}
\label{appsubsec:haar}
We establish several technical lemmas regarding Haar randomness.
As a preliminary, we show that the input state and POVM elements can be considered pure states, validating the assumption addressed in Appendix \ref{appsubsec:F1F2}.
Then, we introduce new notations for use in the lemmas and proofs.

Before proceeding, we clarify our notations by omitting some subscripts and superscripts for simplicity.
Let the input state $\rho$ and the measurement operator $E$ be defined on the Hilbert space $\mathcal{H}_S\otimes\mathcal{H}_A$.
Here, the Hilbert spaces $\mathcal{H}_S$ and $\mathcal{H}_A$ represent the system and ancilla respectively, with dimensions $\text{dim}(\mathcal{H}_S)=d$ and $\text{dim}(\mathcal{H}_A)=d_{\text{anc}}$.
Following these notations, we derive that it is sufficient to consider the input $\rho$ as a pure state.
The primary objective in this section is to find an upper bound of TVD between the observable distributions for $H_0$ and $H_1$, $\text{TVD}(p_0,\E_\psi p_{1,\psi})$.
Writing the input state as a linear sum of pure states $\rho=\sum_{i=1}^{n}c_i\rho^{(i)}$ with $\sum_{i=1}^{n}c_i=1$, we have the upper bound of TVD as follows:
\begin{align}
    \text{TVD}(p_0,\E_\psi p_{1,\psi})
    &=\text{TVD}\left(\sum_{i=1}^{n}c_ip_0^{(i)},\sum_{i=1}^{n}c_i\E_\psi p_{1,\psi}^{(i)}\right)\\
    &\leq \sum_{i=1}^{n}c_i\text{TVD}\left(p_0^{(i)},\E_\psi p_{1,\psi}^{(i)}\right).
\end{align}
Thus, if we have an upper bound of the TVD between the two hypothesis outputs for any pure input state, the same bound also holds for mixed input states.
This justifies considering $\rho$ as a pure state.
Similarly, any POVM element can be expressed as a linear combination of rank-1 operators~\cite{huang2022experiments}, thus validating our assumption to consider pure POVM elements $E$.

Now, we introduce some notations and basic tools for the main part.
In this section, we aim to derive the upper bound of the variable $X$, written as
\begin{align}
    X&\coleq\frac{\tr(E(U_\psi\otimes I_{\text{anc}})\rho(U_\psi\otimes I_{\text{anc}})^\dagger)}{\tr(E\rho)}-1\\
    &=\frac{(e^{-is}-1)\tr (E\rho  I_\psi )+(e^{is}-1)\tr (\rho E I_\psi )+(2-e^{is}-e^{-is})\tr (E I_\psi \rho  I_\psi )}{\tr (E\rho)},
\end{align}
where we denote
\begin{align}
    I_\psi&=\ketbra{\psi}{\psi}\otimes I_{\text{anc}},\\
    f(E,\rho)&= \frac{\tr(\tr_S(E)\tr_S(\rho))}{\tr(E\rho)},
\end{align}
with $\ket{\psi}\in\mathcal{H}_S$.
We follow the notation
\begin{align}
    \ket{e}&=\sum_{i=1}^{d}\sum_{k=1}^{d_{\text{anc}}} e_{ik}\ket{i}\otimes\ket{k}=\sum_{k=1}^{d_{\text{anc}}} \ket{e_k}\otimes\ket{k},\\
    E_{kl}&=\ketbra{e_k}{e_l},\\
    E&=\ketbra{e}{e}=\sum_{k,l=1}^{d_{\text{anc}}}E_{kl}\otimes\ketbra{k}{l},\\
\end{align}
with the similar ones for $\rho=\ketbra{r}{r}$. 
We will employ the relationships of
\begin{align}
    \tr_S(E)&=\sum_{k,l=1}^{d_{\text{anc}}}\braket{e_l}{e_k}\ketbra{k}{l},\\
    \tr_A(E\rho )&=\sum_{k,l=1}^{d_{\text{anc}}}E_{kl}\rho_{lk},
\end{align}
and Hölder's inequality
\begin{align}
\label{eq:holder}
    \|AB\|_1\leq\|A\|_p\|B\|_q
\end{align}
for $1/p+1/q=1$.
Denoting $F_\sigma$ as a permutation operator in the permutation $\sigma$, for example, $F_{(1,2)}(\ket{i}\otimes \ket{j})=\ket{j}\otimes\ket{i}$, it is known that
\begin{align}
    \label{eq:haar}
    \E_{\psi\sim\text{Haar}(d)}(\ketbra{\psi}{\psi})^{\otimes n}=\frac{1}{d(d+1)\dots(d+n-1)}\sum_{\sigma\in S_n}F_\sigma,
\end{align}
which will also be employed in this section.
Here, $S_n$ is the set of $n$-permutations.

\subsubsection{Proof of Lemma \ref{lemma:X_1moment}}
\label{appsubsubsec:X_1moment}
    We show that for $s<1$,
        \begin{align}
            \E_{\psi}X\geq-\frac{s^2}{d}
        \end{align}
    holds.
    We have the following lemma:
    \begin{lemma}
        \label{lemma:haar_1}
        Let $\ket{\psi}$ be a $d$-dimensional Haar-random state. We have
        \begin{align*}
            \E_\psi\tr (E  I_\psi\rho I_\psi )\geq\frac{\tr (E\rho)}{d(d+1)}.
        \end{align*}
    \end{lemma}
    \begin{proof}
        \begin{align}
            \E_\psi\tr (E  I_\psi \rho I_\psi)
            &=\E_\psi\tr \left(\left(\sum_{k,l=1}^{d_{\text{anc}}}\left(E_{kl}\otimes\ketbra{k}{l}\right)\right) I_\psi\left(\sum_{k,l=1}^{d_{\text{anc}}}\left(\rho_{kl}\otimes\ketbra{k}{l}\right)\right) I_\psi\right)\\
            &=\E_\psi\tr \left(\left(\sum_{k,l=1}^{d_{\text{anc}}}\left(E_{kl}\ketbra{\psi}{\psi}\otimes\ketbra{k}{l}\right)\right)\left(\sum_{k,l=1}^{d_{\text{anc}}}\left(\rho_{kl}\ketbra{\psi}{\psi}\otimes\ketbra{k}{l}\right)\right)\right)\\
            &=\E_\psi\tr \left(\tr_A\left(\left(\sum_{k,l=1}^{d_{\text{anc}}}\left(E_{kl}\ketbra{\psi}{\psi}\otimes\ketbra{k}{l}\right)\right)\left(\sum_{k,l=1}^{d_{\text{anc}}}\left(\rho_{kl}\ketbra{\psi}{\psi}\otimes\ketbra{k}{l}\right)\right)\right)\right)\\
            &=\E_\psi\tr \left(\sum_{k,l=1}^{d_{\text{anc}}}E_{kl}\ketbra{\psi}{\psi}\rho_{lk}\ketbra{\psi}{\psi}\right)\\
            &=\sum_{k,l=1}^{d_{\text{anc}}}\E_\psi\bra{\psi}E_{kl}\ket{\psi}\bra{\psi}\rho_{lk}\ket{\psi}\\
            &=\sum_{k,l=1}^{d_{\text{anc}}}\E_\psi\tr \left((E_{kl}\otimes\rho_{lk})(\ketbra{\psi}{\psi})^{\otimes2}\right)\\
            &=\sum_{k,l=1}^{d_{\text{anc}}}\tr \left((E_{kl}\otimes\rho_{lk})\left(\frac{I+F_{(1,2)}}{d(d+1)}\right)\right)\\
            &=\sum_{k,l=1}^{d_{\text{anc}}}\frac{\tr (E_{kl})\tr (\rho_{lk})+\tr (E_{kl}\rho_{lk})}{d(d+1)}\\
            &=\frac{(\sum_{k,l=1}^{d_{\text{anc}}}\braket{e_l}{e_k}\braket{r_k}{r_l})+\tr (E\rho)}{d(d+1)}\\
            &=\frac{\tr (\sum_{k,l=1}^{d_{\text{anc}}}(\ketbra{e_k}{r_k})(\ketbra{r_l}{e_l}))+\tr (E\rho)}{d(d+1)}\\
            &=\frac{\tr ((\sum_{k=1}^{d_{\text{anc}}}\ketbra{e_k}{r_k})(\sum_{l=1}^{d_{\text{anc}}}\ketbra{e_l}{r_l})^\dagger)+\tr (E\rho)}{d(d+1)}\\
            &\geq\frac{\tr (E\rho)}{d(d+1)}.
        \end{align}
    \end{proof}
    Employing $\E_\psi\ketbra{\psi}{\psi}=I/d$, we obtain
    \begin{align}
        \E_{\psi}X
        &=\frac{1}{\tr (E\rho)}\E_{\psi}[(e^{-is}-1)\tr (E\rho  I_\psi )+(e^{is}-1)\tr (\rho E I_\psi )+(2-e^{is}-e^{-is})\tr (E I_\psi \rho  I_\psi )]\\
        &\geq\frac{1}{\tr (E\rho)}\left((e^{-is}-1)\tr \left(E\rho\frac{I}{d}\otimes I_{\text{anc}}\right)+(e^{is}-1)\tr \left(\rho E\frac{I}{d}\otimes I_{\text{anc}} \right)+(2-e^{is}-e^{-is})\frac{\tr (E\rho)}{d(d+1)}\right)\\
        &=\frac{2(1-\cos s)}{\tr (E\rho)}\left(-\frac{\tr (E\rho)}{d}+\frac{\tr (E\rho)}{d(d+1)}\right)\\
        &=-\frac{2(1-\cos s)}{d+1}\\
        &\geq-\frac{s^2}{d},
    \end{align}
    where the second line follows from Lemma \ref{lemma:haar_1} and the last line follows from $s<1$.

\subsubsection{Proof of Lemma \ref{lemma:X_2moment}}
\label{appsubsubsec:X_2moment}
    We show that for $s<1$,
    \begin{align}
        \E_{\psi}X^2\leq\frac{6s^2(f+1)}{d^2}+\frac{72s^3(f^2+1)}{d^4}
    \end{align}
    holds, where we write $f\equiv f(E,\rho)$ for simplicity.
    We obtain the upper bound of the second moment as follows:
    \begin{align}
        \E_{\psi}X^2
        &=\frac{1}{\tr ^2(E\rho)}\E_{\psi}[(e^{-is}-1)^2\tr ^2(E\rho I_\psi )+(e^{is}-1)^2\tr ^2(\rho E I_\psi )+(2-e^{is}-e^{-is})^2\tr ^2(E I_\psi \rho  I_\psi )\notag\\
        &\quad+2(e^{-is}-1)(e^{is}-1)\tr (E\rho  I_\psi )\tr (\rho E I_\psi )\notag\\
        &\quad+2((e^{-is}-1)\tr (E\rho I_\psi )+(e^{is}-1)\tr (\rho E I_\psi ))(2-e^{is}-e^{-is})\tr (E I_\psi \rho  I_\psi )]\\
        &\leq\frac{1}{\tr ^2(E\rho)}(2(2-2\cos s)|\E_{\psi}\tr ^2(E\rho I_\psi)|+(2-2\cos s)^2\E_{\psi}\tr ^2(E I_\psi\rho I_\psi)\notag\\
        &\quad+2(2-2\cos s)\E_\psi\tr (E\rho I_\psi )\tr (\rho E I_\psi )+4(2-2\cos s)^{3/2}|\E_\psi\tr (E\rho I_\psi )\tr (E  I_\psi \rho  I_\psi )|)\\
        &\leq\frac{1}{\tr^2(E\rho)}(2s^2 |\E_{\psi}\tr^2(E\rho I_\psi)|+s^4\E_{\psi}\tr^2(E I_\psi\rho I_\psi)\notag\\
        &\quad+2s^2\E_\psi\tr (E\rho I_\psi )\tr (\rho E I_\psi )+4s^3|\E_\psi\tr (E\rho I_\psi )\tr (E  I_\psi \rho  I_\psi )|) \label{eq:X2bound}.
    \end{align}
    We now obtain the upper bound of each of the four terms in the RHS.
    
    \begin{lemma}
    \label{lemma:haar_2}
        Let $\ket{\psi}$ be a $d$-dimensional Haar-random state. We have
        \begin{align*}
            |\E_\psi\tr ^2(E\rho  I_\psi)|\leq\frac{\tr(E\rho)\tr(\tr_S(E)\tr_S(\rho))+\tr^2 (E\rho)}{d^2}.
        \end{align*}
    \end{lemma}
    \begin{proof}
        We obtain
        \begin{align}
            |\E_\psi\tr ^2(E\rho  I_\psi )|
            &=\left|\E_\psi\tr ^2\left(\left(\sum_{k,l=1}^{d_{\text{anc}}}E_{kl}\otimes\ketbra{k}{l}\right)\left(\sum_{k,l=1}^{d_{\text{anc}}}\rho_{kl}\otimes\ketbra{k}{l}\right)I_\psi\right)\right|\\
            &=\left|\E_\psi\left(\sum_{k,l=1}^{d_{\text{anc}}}\bra{\psi}E_{kl}\rho_{lk}\ket{\psi}\right)^2\right|\\
            &=\left|\E_\psi\left(\bra{\psi}\tr_A(E\rho)\ket{\psi}\right)^2\right|\\
            &=\left|\E_\psi\tr((\tr_A(E\rho)\otimes \tr_A(E\rho))(\ketbra{\psi}{\psi})^{\otimes2})\right|\\
            &=\left|\tr \left((\tr_A(E\rho)\otimes \tr_A(E\rho))\left(\frac{I+F_{(1,2)}}{d(d+1)}\right)\right)\right|\\
            &=\frac{\left|\tr(\tr_A^2(E\rho))+\tr^2(\tr_A(E\rho))\right|}{d(d+1)}\label{eq:haar_2_1}\\
            &=\frac{\left|\tr(\tr_A^2(E\rho))+\tr^2(E\rho)\right|}{d(d+1)}\\
            &\leq\frac{\left|\tr(\tr_A^2(E\rho))\right|+\tr^2(E\rho)}{d(d+1)}\\
            &\leq\frac{\left|\tr(\tr_A^2(E\rho))\right|+\tr^2(E\rho)}{d^2}.
        \end{align}
        The following inequality of
        \begin{align}
            \left|\tr (\tr_A^2(E\rho ))\right|
            &\leq\tr (\tr_A(E\rho )\tr_A(E\rho)^\dagger)\label{eq:haar_2_2}\\
            &=\tr (\tr_A(E\rho )\tr_A(\rho E))\\
            &=\braket{e}{r}\braket{r}{e}\tr (\tr_A(\ketbra{e}{r})\tr_A(\ketbra{r}{e}))\label{eq:haar_2_4}\\
            &=\tr (E\rho )\tr \left(\left(\sum_{k=1}^{d_{\text{anc}}}\ketbra{e_k}{r_k}\right)\left(\sum_{l=1}^{d_{\text{anc}}}\ketbra{r_l}{e_l}\right)\right)\\
            &=\tr (E\rho)\sum_{k,l=1}^{d_{\text{anc}}}\braket{e_l}{e_k}\braket{r_k}{r_l}\\
            &=\tr (E\rho)\tr (\tr_S(E)\tr_S(\rho))\label{eq:haar_2_3}
        \end{align}
        completes the proof.
    \end{proof}

    \begin{lemma}
    \label{lemma:haar_3}
        Let $\ket{\psi}$ be a $d$-dimensional Haar-random state. We have
        \begin{align*}
            \E_\psi\tr(E\rho I_\psi)\tr(\rho E I_\psi)\leq\frac{\tr(E\rho)\tr(\tr_S(E)\tr_S(\rho))+\tr^2 (E\rho)}{d^2}.
        \end{align*}
    \end{lemma}
    \begin{proof}
        The proof is almost the same as that of Lemma \ref{lemma:haar_2}.
        We obtain
        \begin{align}
            \E_\psi\tr (E\rho I_\psi)\tr(\rho EI_\psi)
            &=\frac{\tr(\tr_A(E\rho)\tr_A(\rho E))+\tr(\tr_A(E\rho))\tr(\tr_A(\rho E))}{d(d+1)}\\
            &=\frac{\tr(\tr_A(E\rho)\tr_A(E\rho)^\dagger)+\tr^2(E\rho)}{d(d+1)}\\
            &\leq\frac{\tr(\tr_A(E\rho)\tr_A(E\rho)^\dagger)+\tr^2(E\rho)}{d^2},
        \end{align}
        where the first line is obtained similarly with Eq.~\eqref{eq:haar_2_1}.
        From Eqs.~\eqref{eq:haar_2_2} and \eqref{eq:haar_2_3}, we have
        \begin{align}
            \tr (\tr_A(E\rho)\tr_A(E\rho)^\dagger)\leq\tr (E\rho)\tr (\tr_S(E)\tr_S(\rho)),
        \end{align}
        which completes the proof.
    \end{proof}
    \begin{lemma}
    \label{lemma:haar_4}
        Let $\ket{\psi}$ be a $d$-dimensional Haar-random state.
        We have
        \begin{align*}
            \E_\psi\tr^2(E I_\psi\rho I_\psi)\leq\frac{24(\tr^2(\tr_S(E)\tr_S(\rho))+\tr^2(E\rho))}{d^4}.
        \end{align*}
    \end{lemma}
    \begin{proof}
        Since
        \begin{align}
            \tr (E I_\psi\rho I_\psi)
            &=|\bra{e}(\ketbra{\psi}{\psi}\otimes I_{\text{anc}})\ket{r}|^2\\
            &=\left|\sum_{k=1}^{d_{\text{anc}}}\braket{e_k}{\psi}\braket{\psi}{r_k}\right|^2\\
            &=|\bra{\psi}\tr_A(\ketbra{r}{e})\ket{\psi}|^2,
        \end{align}
        we need to find the upper bound of
        \begin{align}
            \E_\psi\tr^2(E I_\psi\rho I_\psi)=\E_\psi|\bra{\psi}\tr_A(\ketbra{r}{e})\ket{\psi}|^4.
        \end{align}
        Denoting $M\coleq \tr_A(\ketbra{r}{e})$, we have
        \begin{align}
            \E_\psi|\bra{\psi}\tr_A(\ketbra{r}{e})\ket{\psi}|^4
            &\equiv\E_\psi|\bra{\psi}M\ket{\psi}|^4\\
            &=\E_\psi\tr ((M^{\otimes2}\otimes M^{\dagger\otimes2})(\ketbra{\psi}{\psi})^{\otimes 4})\\
            &=\frac{1}{d(d+1)(d+2)(d+3)}\tr \left((M^{\otimes2}\otimes M^{\dagger\otimes2})\sum_{\sigma\in S_4}F_\sigma\right).
        \end{align}
        The trace term in the RHS is a sum of 24 terms, where each term is a multiplication of a trace of a multiplication of $M$'s and $M^\dagger$'s, for instance, $\tr (M^2M^\dagger)\tr (M^\dagger)$ or $\tr (M)\tr (M)\tr (M^{2\dagger})$.
        Employing the inequalities $|\tr (X)|=|\tr (X^\dagger)|\leq\|X\|_1$, $\|X\|_2=\|X^\dagger\|_2\leq\|X\|_1$, and $\|X^\dagger X\|_1\leq\|X\|_2^2$ from Eq.~\eqref{eq:holder} for an arbitrary operator $X$, we can bound each of the 24 terms with $\|M\|_2^m|\tr(M)|^n$ for nonnegative integers $m$ and $n$ satisfying $m+n=4$. 
        For instance, in the case of the first example,
        \begin{align}
            |\tr (M^2M^\dagger)\tr(M^\dagger)|
            &\leq\|M^2M^\dagger\|_1|\tr(M)|\\
            &\leq\|M^2\|_2\|M\|_2|\tr(M)|\\
            &\leq\|M^2\|_1\|M\|_2|\tr(M)|\\
            &\leq\|M\|_2^3|\tr(M)|
        \end{align}
        holds.
        Thus, from $\|M\|_2^m|\tr(M)|^n\leq\|M\|_2^4+|\tr(M)|^4$, we have
        \begin{align}
            \E_\psi|\bra{\psi}M\ket{\psi}|^4
            &\leq\frac{24(\|M\|_2^4+|\tr(M)|^4)}{d(d+1)(d+2)(d+3)}\\
            &\leq\frac{24(\|M\|_2^4+|\tr(M)|^4)}{d^4}\\
            &=\frac{24(\tr^2(M^\dagger M)+|\tr(M)|^4)}{d^4}\\
            &=\frac{24(\tr^2(\tr_A(\ketbra{e}{r})\tr_A(\ketbra{r}{e}))+|\tr(\ketbra{e}{r})|^4)}{d^4}\\
            &=\frac{24(\tr^2(\tr_S(E)\tr_S(\rho))+\tr^2(E\rho))}{d^4},
                && \text{Eqs.}~\eqref{eq:haar_2_4},\eqref{eq:haar_2_3}
        \end{align}
        which completes the proof.
    \end{proof}
    \noindent From Lemma \ref{lemma:haar_2}, \ref{lemma:haar_3}, and \ref{lemma:haar_4}, we have
    \begin{align}
        |\E_\psi\tr ^2(E\rho I_\psi )|
        &\leq\frac{(f+1)\tr^2(E\rho)}{d^2},\label{eq:X_2moment_1}\\
        \E_\psi\tr(E\rho I_\psi)\tr(\rho E I_\psi)
        &\leq\frac{(f+1)\tr^2(E\rho)}{d^2},\label{eq:X_2moment_2}\\
        \E_\psi\tr^2(EI_\psi\rho I_\psi)
        &\leq\frac{24(f^2+1)\tr^2(E\rho)}{d^4}.\label{eq:X_2moment_3}
    \end{align}
    Consequently, we have
    \begin{align}
        s^3|\E_\psi\tr(E\rho I_\psi )\tr(EI_\psi\rho I_\psi)|
        &\leq\E_\psi|s\tr(E\rho I_\psi)s^2\tr(EI_\psi\rho I_\psi)|\\
        &\leq\frac{1}{2}\E_\psi|s\tr(E\rho I_\psi)|^2+\frac{1}{2}\E_\psi|s^2\tr(EI_\psi\rho I_\psi)|^2\\
        &=\frac{s^2}{2}\E_\psi\tr(E\rho I_\psi)\tr(\rho EI_\psi)+\frac{s^4}{2}\E_\psi\tr^2(E  I_\psi\rho I_\psi)\\
        &\leq\frac{s^2(f+1)\tr^2(E\rho)}{2d^2}+\frac{12s^4(f^2+1)\tr^2(E\rho)}{d^4}.\label{eq:X_2moment_4}
    \end{align}

    We now derive the upper bound of Eq.~\eqref{eq:X2bound}.
    From Eqs.~\eqref{eq:X_2moment_1}, \eqref{eq:X_2moment_2}, \eqref{eq:X_2moment_3}, and \eqref{eq:X_2moment_4} with $s<1$, we obtain
    \begin{align}
        &\frac{1}{\tr^2(E\rho)}(2s^2 |\E_{\psi}\tr^2(E\rho I_\psi)|+s^4\E_{\psi}\tr^2(E I_\psi\rho I_\psi)\notag\\
        &\quad+2s^2\E_\psi\tr (E\rho I_\psi )\tr (\rho E I_\psi )+4s^3|\E_\psi\tr (E\rho I_\psi )\tr (E  I_\psi \rho  I_\psi )|)\\
        &\leq\frac{2s^2(f+1)}{d^2}+\frac{24s^4(f^2+1)}{d^4}+\frac{2s^2(f+1)}{d^2}+\frac{2s^2(f+1)}{d^2}+\frac{48s^4(f^2+1)}{d^4}\\
        &\leq\frac{6s^2(f+1)}{d^2}+\frac{72s^4(f^2+1)}{d^4},
    \end{align}
    which completes the proof.
\endgroup

\end{document}